\documentclass[11pt]{article}

\usepackage[top=1in,bottom=1in,left=1in,right=1in]{geometry}
\usepackage[T1]{fontenc} 
\usepackage[noEucal]{main}
\usepackage{aas_macros}
\usepackage{graphicx}
\usepackage{amsthm}
\usepackage{cleveref}
\usepackage{amsmath}
\usepackage{amsfonts}
\usepackage{amssymb}
\usepackage{mathtools}
\usepackage[mathscr]{eucal}
\usepackage{bbold}
\usepackage{braket}
\usepackage{color}
\usepackage{dsfont}
\usepackage{framed}
\usepackage{mathtools}
\usepackage{physics}
\usepackage{centernot}
\usepackage{enumitem} 
\usepackage[a]{esvect}  
\usepackage[scaled=1]{beramono}
\usepackage{slashed}
\usepackage{yfonts}
\usepackage{float}
\usepackage{tikz}
\newcommand*\circled[1]{\tikz[baseline=(char.base)]{
            \node[shape=circle,draw,inner sep=2pt] (char) {#1};}}

\usepackage{subcaption}
\usepackage{float}
\usepackage{afterpage}
\newtheorem{theorem}{Theorem}
\newtheorem{prop}{Proposition}
\newtheorem{defn}{Definition}

\newtheorem{conj}{Conjecture}

\def\[{\big[}
\def\]{\big]}

\def\min{\text{min}}

\def\mfC{{\mathfrak{C}}}

\def\N{{\sf{N}}}

\def\H{{\texttt{H}}}
\def\I{{\texttt{I}}}
\def\S{{\texttt{S}}}

\def\pI{I}
\def\pJ{J}
\def\pK{K}

\def\pII{{\underline{I}}}
\def\pJJ{{\underline{J}}}
\def\pKK{{\underline{K}}}
\def\pLL{{\underline{L}}}
\def\Res{{\text{Res}}}

\renewcommand{\oslash}{\emptyset}

\begin{document}
\begin{titlepage}
\unitlength = 1mm

\hfill CALT-TH 2023-052

\hfill MIT-CTP/5659

\vskip 3cm
\begin{center}

{\huge{Beyond the Holographic Entropy Cone via Cycle Flows}}

\vspace{0.8cm}
Temple He$^1$, Sergio Hern\'{a}ndez-Cuenca$^2$, Cynthia Keeler$^3$

\vspace{1cm}

{\it  $^1$Walter Burke Institute for Theoretical Physics, California Institute of Technology, \\ Pasadena, CA 91125 USA}\\
{\it  $^2$Center for Theoretical Physics, Massachusetts Institute of Technology, \\ Cambridge, MA 02139 USA} \\
{\it $^3$Department of Physics, Arizona State University, \\ Tempe, AZ 85281 USA}

\vspace{0.8cm}

\begin{abstract}

Motivated by bit threads, we introduce a new prescription for computing entropy vectors outside the holographic entropy cone. By utilizing cycle flows on directed graphs, we show that the maximum cycle flow associated to any subset of vertices, which corresponds to a subsystem, manifestly obeys purification symmetry. Furthermore, by restricting ourselves to a subclass of directed graphs, we prove that the maximum cycle flow obeys both subadditivity and strong subadditivity, thereby establishing it as a viable candidate for the entropy associated to the subsystem. Finally, we demonstrate how our model generalizes the entropy vectors obtainable via conventional flows in undirected graphs, as well as conjecture that our model similarly generalizes the entropy vectors arising from hypergraphs. 

\end{abstract}

\vspace{1.0cm}
\end{center}
\end{titlepage}
\pagestyle{empty}
\pagestyle{plain}
\pagenumbering{arabic}

\tableofcontents

\normalem

\section{Introduction}

The analysis of the holographic entropy cone (HEC) was initiated in \cite{Bao:2015bfa}, and its success has relied crucially on the fact that the entanglement entropy of the boundary conformal field theory (CFT) is captured geometrically by Ryu-Takayanagi (RT) surfaces in the bulk for static geometries \cite{Ryu:2006bv}.\footnote{The covariant generalization of RT surfaces is given in \cite{Hubeny:2007xt}.} In particular, given any boundary region $A$ of the CFT, the RT surface associated to $A$ is the surface homologous to $A$ in the bulk anti-de Sitter (AdS) spacetime with the minimum area. The entanglement entropy between $A$ and its complement on the boundary is then proportional to the area of the RT surface. This can of course be easily generalized to computing the entanglement entropy between multiple disjoint subregions on the boundary. We will henceforth always implicitly assume that the subregions are disjoint.

A crucial insight made in \cite{Bao:2015bfa} was that the geometrization of entanglement entropy for CFTs implies that the structure of entanglement entropy for the various subregions can be described in terms of a discrete graph. More precisely, suppose we are interested in the entanglement entropy of the boundary subregions $A_1,A_2,\ldots A_{\N}$, and denote $A_{\N+1} = (A_1 \cup A_2 \cup \dots A_\N)^c$ as the purifier, where the superscript $c$ denotes the complement. Then there exists a discrete undirected graph with a subset of its vertices labeled as $A_1,A_2,\ldots,A_{\N+1}$, which we call boundary vertices, such that the min cut of any subset of the boundary vertices gives the entanglement entropy associated to the union of corresponding boundary subregions. We will call all other vertices internal vertices.

The recasting of holographic entropies into min cuts of discrete graphs was enormously fruitful, as it allowed for the utilization of various graph theoretic tools \cite{He:2019ttu, He:2020xuo,Avis:2021xnz,Czech:2021rxe,Fadel:2021urx,Hernandez-Cuenca:2022pst,Hernandez-Cuenca:2023iqh}. For instance, the max flow-min cut theorem was used to give a ``bit thread'' interpretation of holographic entropies \cite{Freedman:2016zud, Headrick:2017ucz, Cui:2018dyq, Harper:2019lff, Headrick:2020gyq, Headrick:2022nbe}. Under this perspective, rather than thinking of entanglement entropy as given by a minimum cut in the graph, it is given by a maximal flow between a subset of boundary vertices and its complement. In the original boundary CFT, we can interpret the entanglement entropy of a subregion as the maximum flux of a vector field that flows through the AdS bulk geometry out of the subregion.

Despite the various successes of studying entropies of holographic states, there are many quantum states, such as the 4-party GHZ state, that are not holographic. It is then a natural question if there is some generalization of undirected graphs that can perhaps capture a wider class of states. The most obvious choice may be to allow the graph to have directed edges. However, if the entropy associated to a subset of boundary vertices is given by the max flow, then there is no guarantee that it is equal to the max flow out of the complement of the subset of boundary vertices, thereby violating purification symmetry of entanglement entropy. Nevertheless, rather than introducing directed edges, there have been various such generalizations in the literature, including both hypergraphs \cite{Bao:2020zgx} and links \cite{Bao:2021gzu}.%
\footnote{In particular, stabilizer states provide a well-studied class of states which are not all holographic. The papers cited above use the stabilizer states as an example. For more on the entropies available to stabilizer states see \cite{Bao:2020zgx,Walter:2020zvt,Bao:2020mqq,Keeler:2022ajf,Keeler:2023xcx,Keeler:2023shl}; for similar explorations on Dicke states, see \cite{Munizzi:2023ihc}.} Such generalizations automatically respect purification symmetry as there is no directionality in their structures, while also allowing for min cut structures that are disallowed in the original undirected graphs, and there has been fruitful progress in using them to explore entropy vectors of non-holographic states \cite{Walter:2020zvt, Bao:2020mqq, He:2023cco, He:2023aif}. 

In this paper, we present an alternative way of moving beyond entropies of holographic states by using instead another notion of flows. Simply stated, our protocol to compute the entropy $\S(\pII)$ from a graph is the following.
\begin{itemize}
    \item Begin with a balanced directed network, which are graphs with directed edges obeying the condition where at every vertex the total weight of ingoing edges equals that of outgoing edges, that also obeys a property known as ``nesting'' (introduced below).\footnote{For us, we choose to call a network a graph with edge weights. This is defined precisely later.}
    \item Identify the set of boundary vertices $\p V$ which represent boundary regions.
    \item Partition the boundary vertices $\p V$ into the set $\pII$ and its complement, $\pII^c \coloneqq \p V \setminus \pII$.\footnote{We underline $\pII$ to indicate it may contain the boundary vertex corresponding to the purifier.}
    \item Consider closed paths (cycles) in the graph which pass through at least one vertex in $\pII$ and at least one vertex in $\pII^c$.
    \item The entropy $\S(\pII)$ is given by the magnitude of the maximum simultaneous cycle flow, defined in \Cref{ssec:graph-theory}, along all such paths.
\end{itemize}
Note that this protocol uses cycle flows, which are flows with neither sources nor sinks, on a subclass of balanced directed networks. Additionally, the specific flow pattern is important; we need to know which cycles have flow turned on, not just which edges.  As we shall see, this generalization to directed graphs under a cycle flow protocol allows access to a broader set of entropy vectors than the undirected or hyperedge graph protocols.  Additionally, the cycle flow protocol retains some of the topological features of the link picture.

Importantly, this cycle flow protocol respects purification symmetry, since the vertex set $\pII$ and its complement $\pII^c$ play the same role.  As we show below, entropies calculated from the cycle flow protocol always obey both subadditivity (SA) and strong subadditivity (SSA), which are necessary conditions for identifying the flow as entanglement entropy. Furthermore, we will also show that in general, the flow will not satisfy monogamy of mutual information (MMI), which is a holographic entropy inequality \cite{Hayden:2011ag}, thereby implying entropies represented by such cycle flows represent non-holographic states as well.

Our paper is structured as follows. In \Cref{sec:prelim}, we give a brief review of relevant background information from both quantum information theory and graph theory. We will also use the opportunity to establish the various conventions used throughout. In \Cref{ssec:entropy}, we will formally introduce our prescription for identifying the maximum cycle flow with the entropy. We then prove in \Cref{ssec:sa} and \ref{ssec:ssa} that our prescription satisfies both SA and SSA, giving evidence that such maximum cycle flows can indeed be thought of as von Neumann entropies. Finally, we conclude with a discussion and open questions in \Cref{sec:discussion}. In \Cref{app:dual-program}, we remark on the dual program and how to prove SA using it.

\section{Preliminaries and Conventions}\label{sec:prelim}

In this section, we review relevant background information from both graph theory and quantum information theory. The quantum information theory portion is reviewed in \Cref{ssec:qi}, and the graph theory portion is reviewed in \Cref{ssec:graph-theory}.

\subsection{Quantum Information Theory}\label{ssec:qi}

The systematic study of the entropy cone associated to a general quantum system is organized in terms of the number of parties.\footnote{For the holographic entropy cone, each party corresponds to a spatial boundary subregion of the CFT, and all the subregions are disjoint.} If there are $\N$ parties $A_1,A_2,\ldots,A_\N$, we can specify the entanglement structure of the system by a vector in $2^{\N}-1$ dimensions, with each component being $\S(\pI)$, the entropy of the subsystem $\bigcup_{i \in \pI} A_{i}$, for every nonempty subset $\pI \subseteq[\N] \coloneqq \{1,2,\ldots,\N\}$. For small number of parties, we can also label the parties $A,B,\ldots$. The vector space in which this entropy vector lives is called the entropy space. For example, for $\N=3$ parties, with the parties labeled as $A,B$, and $C$, the entropy vector lives in the $7$-dimensional entropy space and is given by
\begin{align}
    \vec \S = \big(\S(A), \; \S(B), \; \S(C), \; \S(AB), \; \S(AC), \; \S(BC), \; \S(ABC) \big).
\end{align}
In general, we will label the components of the entropy vector in lexicographical order. There is the additional purifier party $A_{\N+1} \coloneqq (A_1 \cup A_2 \cup \cdots \cup A_\N)^c$ (or $O$ if we are labeling the parties using the alphabet), but by purification symmetry we have
\begin{align}\label{eq:purification}
    \S(\pII) = \S(\pII^c),
\end{align}
for any $\pII \subseteq [\N+1]$ and $\pII^c \coloneqq [\N+1]  \setminus \pII$ is the complement. We will always use $\pI,\pJ,\ldots$ to denote subsets of $[\N]$, and $\pII,\pJJ,\ldots$ to denote subsets of $[\N+1]$, which includes the purifier. Thus, the entropy of any subset of parties involving the purifier $A_{\N+1}$ can be rewritten as the entropy of its complement, which does not involve $A_{\N+1}$ and is hence a component of the entropy vector as defined above.

Not every point in entropy space is realizable by quantum states. For instance, the positivity of entropy implies that for all $\pII \subseteq [\N+1]$, we have
\begin{align}
    \S(\pII) \geq 0.
\end{align}
This means that entropy vectors corresponding to some quantum state must live in the region of entropy space where all its components are positive. Furthermore, we also know that entropies obey certain linear inequalities, which for $\pI,\pJ,\pK \subseteq [\N]$ are given by
\begin{align}
    \text{Subadditivity (SA):}\quad & \S(\pI) + \S(\pJ) \geq \S(\pI\pJ) \\
    \text{Araki-Lieb (AL):}\quad & \S(\pI) + \S(\pI\pJ) \geq \S(\pJ) \\
    \text{Strong subadditivity (SSA):} \quad & \S(\pI\pJ) + \S(\pJ\pK) \geq \S(\pJ) + \S(\pI\pJ\pK) \\
    \text{Weak monotonicity (WM):} \quad & \S(\pI\pJ) + \S(\pJ\pK) \geq \S(\pI) + \S(\pK),
\end{align}
where $\pI\pJ \coloneqq \pI \cup \pJ$ is shorthand for the union of the two subset of regions, and $\pI\pJ\pK$ is similarly defined for three regions. In the above inequalities, we did not utilize the purifier entropy $\S(\N+1)$. It is an easy exercise to see that AL and WM are related to SA and SSA by purification symmetry \eqref{eq:purification}, respectively. Therefore, we can consolidate the four types of inequalities above into the following two types:
\begin{align}
    \text{Subadditivity (SA):}\quad & \S(\pII) + \S(\pJJ) \geq \S(\pII\pJJ) \label{eq:sa} \\
    \text{Strong subadditivity (SSA):} \quad & \S(\pII\pJJ) + \S(\pJJ\pKK) \geq \S(\pJJ) + \S(\pII\pJJ\pKK) . \label{eq:ssa}
\end{align}
Henceforth, SA and SSA will refer to the inequalities \eqref{eq:sa} and \eqref{eq:ssa}, which can involve subsystems containing the purifier $\N+1$.

For any choice of $\pII,\pJJ,\pKK$, the inequalities \eqref{eq:sa} and \eqref{eq:ssa} each define a half-space in entropy space. The intersection of all such half-spaces then forms a polyhedral cone in which entropy vectors corresponding to realizable quantum states live. We will refer to this cone as the SSA entropy cone (even though the cone is constrained by both SA and SSA). The quantum entropy cone (QEC) is defined to be the topological closure of the set of all the entropy vectors realizable by quantum states, and was shown in \cite{pippenger2003inequalities} 
to be a convex cone. Thus, because all such entropy vectors must obey SA and SSA, the QEC lives inside the SSA entropy cone. It is known that for $\N=3$ the QEC actually coincides with the SSA entropy cone. However, for $\N\geq 4$, besides the existence of infinitely many constrained inequalities (inequalities that only apply on subspaces of codimension higher than zero) \cite{Linden:2004ebt,Cadney_2012}, very little is known about the QEC.\footnote{In particular, it remains an open question whether the QEC is even polyhedral for higher parties.}

Holographic states are a subset of quantum states. Therefore, the holographic entropy cone (HEC), which is a convex, polyhedral cone consisting of the set of all entropy vectors realizable by holographic states \cite{Bao:2015bfa,Avis:2021xnz}, lies within the QEC. Indeed, we know that holographic states satisfy additional inequalities that are not necessarily satisfied by all quantum states. These are known as holographic entropy inequalities (HEIs), and one such inequality is the monogamy of mutual information (MMI) \cite{Hayden:2011ag,Cui:2018dyq}, which is given by
\begin{align}\label{eq:mmi}
\begin{split}
    &\text{Monogamy of mutual information (MMI):} \\
    &\qquad\qquad \S(\pII\pJJ) + \S(\pII\pKK) + \S(\pJJ\pKK) \geq \S(\pII) + \S(\pJJ) + \S(\pKK) + \S(\pII\pJJ\pKK) .
\end{split}
\end{align}
Notice that MMI and SA together imply SSA, since adding \eqref{eq:sa} to \eqref{eq:mmi} yields (a relabeling of) \eqref{eq:ssa}. MMI is the only HEI for $\N < 5$, but additional HEIs are known for $\N \geq 5$ \cite{Bao:2015bfa,Cuenca:2019uzx,Czech:2022fzb,Hernandez-Cuenca:2023iqh,Czech:2023xed}.

It is known that a model for describing entropy vectors in the HEC uses undirected graphs \cite{Bao:2015bfa}. More specifically, given any entropy vector in HEC involving $\N$ parties (excluding the purifier), there exists an undirected graph $\CG=(V,E)$ such that a subset of vertices $\p V \coloneqq \{1,2,\ldots,\N+1\} \subseteq V$, which we call boundary vertices, correspond to the $\N$ parties plus the purifier, and the min cut of any subset of boundary vertices equals the entropy associated to the corresponding subset of parties. Entropy as defined by such graph models satisfy the HEIs by construction. Because our goal is to explore beyond the HEC (but still remain within the QEC), we would like to determine an alternative model to describe a class of entropy vectors that satisfy SA \eqref{eq:sa} and SSA \eqref{eq:ssa} without having to satisfy HEIs like MMI \eqref{eq:mmi}. We will introduce such a model in \Cref{sec:entropy-cone}, but before continuing, let us first review some preliminaries from graph theory in the next subsection.

\subsection{Graph Theory}\label{ssec:graph-theory}

We begin by introducing notation and establishing conventions for relevant graph theory concepts. Consider first a discrete directed graph $\CG \coloneqq (V,E)$, where $V$ is the collection of vertices and $E$ is the collection of directed edges such that there is at most one directed edge from any vertex $v \in V$ to any other vertex $v' \in V$.\footnote{Note that undirected graphs are a special case of directed graphs where between every two vertices $v_1, v_2 \in V$ there exists a directed edge from $v_1$ to $v_2$ and another directed edge from $v_2$ to $v_1$ (see \Cref{fig:GtoN}). Thus, it suffices for us to study directed graphs only.} Associated to every directed edge $e \in E$ is a capacity function $c:\, E \to \mathbb R_{\geq 0}$, which maps every edge to a non-negative real number known as the edge's capacity or weight. We can then define a network $\CN \coloneqq (V,E,c)$ to be a graph $(V,E)$ with a capacity function $c$.

For every directed edge $e$, it is useful to define $s(e) \in V$ to be the source vertex from which the edge originates from, and $t(e) \in V$ to be the target vertex towards which the edge is directed. Using the capacity function, we will now introduce the following notion of simple cycle flows, which will be useful for our later construction of multi-cycle flows.

\begin{defn}[Simple cycle flow] \label{def:cycle-flow}
    Given a network $\CN = (V,E,c)$, we define a edge cycle to be an ordered sequence of distinct edges $E_\CC \coloneqq \{e_1, e_2, \ldots, e_n \}$ such that  $t(e_i) = s(e_{i+1})$ for $i=1,\ldots,n$ with $e_{n+1} \coloneqq e_1$. Notice that $E_\CC$ is defined only up to cyclic permutations. A simple cycle flow is then an ordered pair $\CC \coloneqq (E_\CC; \l)$, where $E_\CC$ is any edge cycle and $\l \in \mathbb R_{\geq 0}$ is a non-negative number,
    satisfying the property   
    \begin{align}
        c(e) \geq \l \quad\text{for every $e \in E_\CC$.}
    \end{align}
    We refer to $\l$ as the magnitude,  or flux, of the simple cycle flow $\CC$.
\end{defn}

Unlike typical flows in network theory involving sources and sinks, which have arisen in quantum information literature as bit threads \cite{Freedman:2016zud}, simple cycle flows have not been connected to any notion of entropy. However, we will introduce a proposal in the next section that relates cycle flows to entanglement entropy beyond holographic states. It would be interesting to see if there are any relations between the cycle flows we are considering and the more well-understood bit threads (see \Cref{fn:bit-thread-connection}).\footnote{For an exploration of the more conventional flows involving sources and sinks, and how they pertain to the study of the holographic entropy cone in the context of bit threads, we refer the reader to \cite{Freedman:2016zud} and in particular Section 6 of \cite{Cui:2018dyq}.}

In order to connect cycle flows to entropy, we will need to define a notion of magnitude for such flows. Typically in network theory, the magnitude of a flow is measured by the flux from vertices that are sources to vertices that are sinks. However, cycle flows have zero flux, as there are no vertices that are sources or sinks. Hence, the flux of a cycle flow into any vertex equals that out of the vertex. Nevertheless, given the set of boundary vertices $\p V = [\N+1] \subseteq V$, we can introduce an alternative notion of the magnitude of a flow through some subset of boundary vertices. As was introduced in \Cref{ssec:qi}, subsets of boundary vertices will be denoted using $\pII,\pJJ$, and $\pKK$.

\begin{defn} \label{def:flux}
    Given a network $\CN = (V,E,c)$, let $\p V = [\N+1] \subseteq V$ be a subset of vertices called boundary vertices. For any simple cycle flow $\CC = (E_\CC;\lambda)$ on $\CN$, define 
    \begin{align}
        V(\CC) \coloneqq \bigg( \bigcup_{e \in E_\CC} t(e) \bigg) \cap \p V   
    \end{align}
    to be the boundary vertices part of the cycle flow $\CC$. Letting $\pII,\pJJ \subseteq \p V$ denote any two disjoint subsets of boundary vertices, we define the function
    \begin{align}\label{eq:single-flux}
        f_\CC(\pII,\pJJ) = \begin{cases}
            \lambda & V(\CC) \cap \pII \neq \oslash \quad\text{and}\quad V(\CC) \cap \pJJ \neq \oslash \\
            0 & \text{otherwise.}
        \end{cases}
    \end{align}
    Notice that $f_\CC$ is symmetric with respect to its two arguments, and we call $f_\CC(\pII,\pJJ)$ the flux between $\pII$ and $\pJJ$ due to simple cycle $\CC$. Thus, the flux between $\pII$ and $\pJJ$ is nonzero only if $\CC$ involves at least one vertex in $\pII$ and one in $\pJJ$.
\end{defn}

Utilizing simple cycles, we can now define the analog of multi-flows for cycle flows, which we call multi-cycle flows.

\begin{defn}[Multi-cycle flow] \label{def:multi-flow}
    Given a network $\CN = (V,E,c)$, a multi-cycle flow is a set of simple cycle flows $\mfC \coloneqq \{\CC_i = (E_{i};\l_i)\}_{i=1,\ldots,m}$ obeying the following condition. For every $e \in E$, defining $u_i(e) = 1$ if $e \in E_{i}$ and 0 otherwise, we have
    \begin{align} \label{eq:multi-flow-cond}
        \sum_{i=1}^m u_i(e) \l_i \leq c(e).
    \end{align}
    Furthermore, the set of boundary vertices which are part of the multi-cycle flow is given by
    \begin{align}
        V(\mfC) \coloneqq \bigcup_{\CC \in \mfC} V(\CC).
    \end{align}
    
    Lastly, letting $\pII,\pJJ \subseteq \p V$ denote any two disjoint subsets of boundary vertices, the flux between $\pII$ and $\pJJ$ due to $\mfC$ is defined to be
    \begin{align}
        f_\mfC(\pII,\pJJ) = \sum_{i=1}^m f_{\CC_i}(\pII,\pJJ) ,
    \end{align}
    where $f_{\CC_i}(\pII,\pJJ)$ is the flux between $\pII$ and $\pJJ$ due to simple cycle flow $\CC_i$, as defined in \eqref{eq:single-flux}.
\end{defn}

Given any network $\CN$ and any subset of boundary vertices $\pII \subseteq \p V$, it is clear that the maximum flux between $\pII$ and its complement $\pII^c \coloneqq \p V \setminus \pII$ is bounded above by the sum of capacities of edges emanating from $\pII$, that is,
\begin{align}
    \max_{\mfC} f_\mfC(\pII,\pII^c) \leq \sum_{e\in E:\, s(e) \in \pII} c(e).
\end{align}
In particular, this means the maximum flux between $\pII$ and $\pII^c$ is well-defined, and we denote it as
\begin{align}\label{eq:max-flow}
    \S(\pII) \coloneqq \max_{\mfC} f_\mfC(\pII,\pII^c) .
\end{align}
In the following section, we will identify the maximum flux of a multi-cycle flow through some boundary vertices $\pII$ of a \emph{balanced} network (defined in \Cref{def:balance}) as the entropy, hence the suggestive notation of denoting it as $\S(\pII)$. Furthermore, for convenience, we will henceforth denote 
\begin{align}
    f_\CC(\pII) \coloneqq f_\CC(\pII,\pII^c) , \qquad f_\mfC(\pII) \coloneqq f_\mfC(\pII,\pII^c).
\end{align}
We will also introduce the notion of a subflow in a multi-cycle flow.
\begin{defn}[Multi-cycle subflow]
    Given a network $\CN = (V,E,c)$ and a multi-cycle flow $\mfC:= \{\CC_i = (E_i;\l_i)\}_{i=1,\ldots,m}$, a multi-cycle $\mfC' = \{\CC_i' = (E_i';\l_i')\}_{i=1,\ldots,m}$ is a multi-cycle subflow of $\mfC$ if $E_i' = E_i$ and $\l_i' \leq l_i$ for all $i$.\footnote{If $\l_j' = 0$ for some $j$ this just means $\mfC'$ does not contain the simple cycle $\CC_j$.} 
\end{defn}

Let us illustrate the notion of flux and multi-cycle flows with a simple example. In the figure below, we draw two possible multi-cycle flows, where $\mfC_1$ is on the left figure and $\mfC_2$ is on the right figure. The multi-cycle flows $\mfC_1$ and $\mfC_2$ each consists of two simple cycles colored by red and blue. We can easily compute
\begin{align}
\begin{split}
    f_{\mfC_1}(A) = 2, \qquad f_{\mfC_1}(B) = 2, \qquad f_{\mfC_1}(AB) = 0, \\
    f_{\mfC_2}(A) = 2, \qquad f_{\mfC_2}(B) = 2, \qquad f_{\mfC_2}(AB) = 1.
\end{split}
\end{align}
In particular, notice that the flux out of $AB$ in $\mfC_1$ differs from that in $\mfC_2$.

\begin{figure}[H] 
    \centering
    \begin{tikzpicture}[scale=1, node distance=2cm]
        \node[style={draw, circle, fill=black, inner sep=2pt}] (A) at (0, 2) {};
        \node[style={draw, circle, fill=black, inner sep=2pt}] (B) at (0, 0) {};
        \node[style={draw, circle, fill=black, inner sep=2pt}] (i) at (1, 1) {};
        \node[style={draw, circle, fill=black, inner sep=2pt}] (O) at (3, 1) {};

        \node[right=2mm] at (O) {$O$};
        \node[below=2mm] at (B) {$B$};
        \node[above=2mm] at (A) {$A$};

        \draw[-latex, line width=1pt, color=red] (A) -- node[left] {} (B);
        \draw[-latex, line width=1pt, color=red] (A) to[bend right=45] node[left] {} (B);
        \draw[-latex, line width=1pt, color=red] (B) -- node[below] {} (i);
        \draw[-latex, line width=1pt, color=red] (B) to[bend right=45] node[below] {} (i);
        \draw[-latex, line width=1pt, color=red] (i) -- node[above] {} (A);
        \draw[-latex, line width=1pt, color=red] (i) to[bend right=45] node[above] {} (A);
        \draw[-latex, line width=1pt, color=blue] (i) to[bend left=45] node[above] {} (O);
        \draw[-latex, line width=1pt,color = blue] (O) to[bend left=45] node[below] {} (i);

        \begin{scope}[xshift=6cm]
            \node[style={draw, circle, fill=black, inner sep=2pt}] (A) at (0, 2) {};
        \node[style={draw, circle, fill=black, inner sep=2pt}] (B) at (0, 0) {};
        \node[style={draw, circle, fill=black, inner sep=2pt}] (i) at (1, 1) {};
        \node[style={draw, circle, fill=black, inner sep=2pt}] (O) at (3, 1) {};

        \node[right=2mm] at (O) {$O$};
        \node[below=2mm] at (B) {$B$};
        \node[above=2mm] at (A) {$A$};

        \draw[-latex, line width=1pt, color=red] (A) -- node[left] {} (B);
        \draw[-latex, line width=1pt, color=blue] (A) to[bend right=45] node[left] {} (B);
        \draw[-latex, line width=1pt, color=red] (B) -- node[below] {} (i);
        \draw[-latex, line width=1pt, color=blue] (B) to[bend right=45] node[below] {} (i);
        \draw[-latex, line width=1pt, color=red] (i) -- node[above] {} (A);
        \draw[-latex, line width=1pt, color=blue] (i) to[bend right=45] node[above] {} (A);
        \draw[-latex, line width=1pt, color=blue] (i) to[bend left=45] node[above] {} (O);
        \draw[-latex, line width=1pt, color=blue] (O) to[bend left=45] node[below] {} (i);
        \end{scope} 
    \end{tikzpicture}
    \caption{All the flows implicitly are weight 1. We illustrate two  multi-cycle flows on the graph. In both cases, the multi-cycle flow consists of two simple cycles, which are colored red and blue in the figure. Notice that the two multi-cycle flows give rise to different fluxes through the subsets of vertices.} \label{fig:example1}
\end{figure}
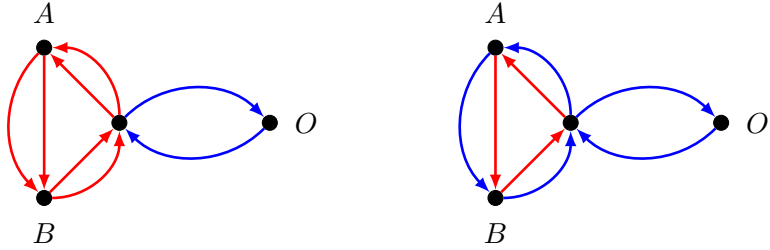

For future convenience, it is actually useful to define a notion of a ``minimal'' simple cycle that does not use unnecessary edges. To make this definition precise, we begin by defining the natural notion of a subcycle of a simple cycle flow.
\begin{defn}[Subcycle flow]
\label{def:subcycle}
     Given a simple cycle flow $\CC =(E_\CC;\lambda)$, where $E_{\CC} = \{e_1,\ldots,e_n\}$, a subcycle flow is any simple cycle flow $\CC' =(E_\CC';\lambda)$, with $E_{\CC'}$ being a subsequence of $E_\CC$.
\end{defn}

From \Cref{def:flux}, it is clear that the flux between two sets of boundary vertices $\pII$ and $\pJJ$ for a simple cycle flow $\CC$ will be unchanged if $\CC$ is replaced by any subcycle flow $\CC'$ so long as the latter also traverses vertices in both $\pII$ and $\pJJ$. For this reason, it makes sense to define a notion of minimality of simple cycle flows relative to vertex subsets as follows.

\begin{defn}[Minimal simple cycle flow]
\label{def:minimal-cycle}
    A simple cycle flow $\CC$ is minimal with respect to a collection of boundary vertex subsets $\{\pII_1,\pII_2,\ldots, \pII_m\}$ if, for any subcycle flow $\CC'$ with $E_{\CC'} \subset E_{\CC}$ , the flux $f_{\CC'}(\pII_i,\pII_j) < f_{\CC}(\pII_i,\pII_j)$ for some $i,j\in[m]$. 
\end{defn}

\begin{figure}
    \centering
\centering
    \begin{tikzpicture}[scale=1, node distance=2cm]
        \node[style={draw, circle, fill=black, inner sep=2pt}] (A2) at (0, 0) {};
        \node[style={draw, circle, fill=black, inner sep=2pt}] (A1) at (4, 0) {};
        \node[style={draw, circle, fill=black, inner sep=2pt}] (B) at (2, 1.5) {};
        \node[style={draw, circle, fill=black, inner sep=2pt}] (C) at (2, 4) {};

        \node[above=2mm] at (B) {$B$};
        \node[below=2mm] at (A1) {$A_1$};
        \node[below=2mm] at (A2) {$A_2$};
        \node[above=2mm] at (C) {$C$};

        \draw[-latex, line width=1pt, color=black] (A2) to[bend left=30] node[below] {} (B);
        \draw[-latex, line width=1pt, color=black] (B) to[bend left=30] node[above] {} (A2);
        \draw[-latex, line width=1pt, color=red] (B) to[bend left=30] node[above] {$\circled{\scriptsize{4}}$} (A1);
        \draw[-latex, line width=1pt,color = red] (A1) to[bend left=30] node[below] {$\circled{\scriptsize{1}}$} (B);
        \draw[-latex, line width=1pt,color = red] (B) to[bend left=30] node[left] {$\circled{\scriptsize{2}}$} (C);
        \draw[-latex, line width=1pt,color = red] (C) to[bend left=30] node[right] {$\circled{\scriptsize{3}}$} (B);

        \begin{scope}[xshift=8cm]
        \node[style={draw, circle, fill=black, inner sep=2pt}] (A2) at (0, 0) {};
        \node[style={draw, circle, fill=black, inner sep=2pt}] (A1) at (4, 0) {};
        \node[style={draw, circle, fill=black, inner sep=2pt}] (B) at (2, 1.5) {};
        \node[style={draw, circle, fill=black, inner sep=2pt}] (C) at (2, 4) {};

        \node[above=2mm] at (B) {$B$};
        \node[below=2mm] at (A1) {$A_1$};
        \node[below=2mm] at (A2) {$A_2$};
        \node[above=2mm] at (C) {$C$};

        \draw[-latex, line width=1pt, color=red] (A2) to[bend left=30] node[above] {$\circled{\scriptsize{3}}$} (B);
        \draw[-latex, line width=1pt, color=red] (B) to[bend left=30] node[below] {$\circled{\scriptsize{2}}$} (A2);
        \draw[-latex, line width=1pt, color=red] (B) to[bend left=30] node[above] {$\circled{\scriptsize{6}}$} (A1);
        \draw[-latex, line width=1pt,color = red] (A1) to[bend left=30] node[below] {$\circled{\scriptsize{1}}$} (B);
        \draw[-latex, line width=1pt,color = red] (B) to[bend left=30] node[left] {$\circled{\scriptsize{4}}$} (C);
        \draw[-latex, line width=1pt,color = red] (C) to[bend left=30] node[right] {$\circled{\scriptsize{5}}$} (B);
        \end{scope} 
    \end{tikzpicture}
    \caption{In the figure, we have a partition of the four vertices of the graph into $\{A_1,A_2\}$, $\{B\}$, and $\{C\}$. All the edges have capacity 1, and the circled numbers indicate the ordering of the edge cycle in the cycle flow. On the left, we have a minimal cycle flow that is colored in red. On the right we have a simple cycle flow that is not a minimal cycle flow, since we can remove the subcycle involving edges labeled by 2 and 3 without changing the flux between any pair of vertex subsets $\{A_1,A_2\}$, $\{B\}$, and $\{C\}$.} \label{fig:minimal-cycle}
\end{figure}
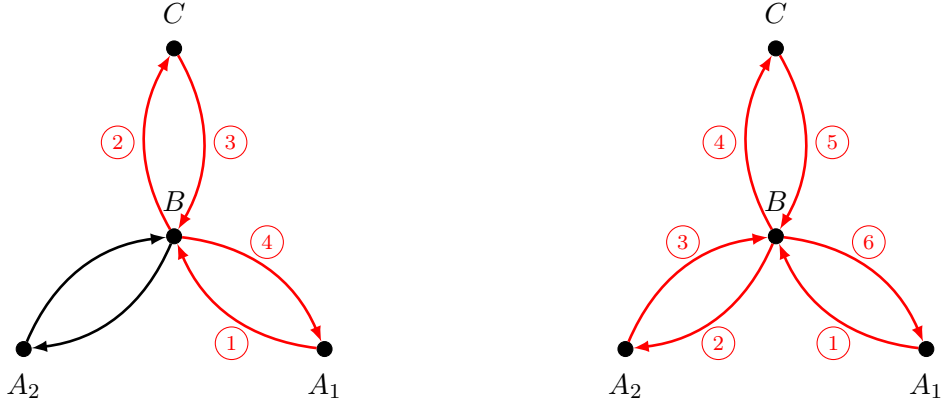

We illustrate the difference between a simple cycle flow that is a minimal cycle flow and one that is not in \Cref{fig:minimal-cycle}. Effectively, minimal cycles allow us to capture the flux between subsets of boundary vertices without ``wasting'' any flows on edges that do not contribute to the flux. Similarly, we can readily generalize the notion of ``minimality'' to multi-cycle flows.

\begin{defn}[Minimal multi-cycle flow]
    Given a partition of the boundary vertices $\p V$ into subsets $\{\pII_1, \ldots,\pII_n\}$, a multi-cycle flow $\mfC$ is minimal with respect to $\{\pII_1,\ldots,\pII_n\}$ if all the simple cycle flows comprising $\mfC$ are minimal with respect to $\{\pII_1,\ldots,\pII_n\}$.
\end{defn}

Finally, we want to begin restricting ourselves to the types of directed networks under consideration. The first restriction we impose is to focus on balanced networks, which are defined as follows.

\begin{defn}\label{def:balance}
    A network $\CN = (V,E,c)$ is balanced if at every vertex $v \in V$, the total capacity of ingoing edges equals that of outgoing edges. In other words, 
    \begin{align}
        \sum_{e \in E:\, s(e)=v} c(e) = \sum_{e \in E:\, t(e) = v} c(e) \quad\text{for all $v \in V$.}
    \end{align}
\end{defn}

The reason we want to restrict ourselves to balanced networks is because otherwise, SSA fails if the entropy through some boundary vertices $\pII$ is given by the maximum flux through it. For instance, consider the unbalanced network below, where every edge has capacity 1. It is easy to verify that the maximum flux through $C$ is 2, through $AC$ is 1, through $BC$ is 2, and through $ABC$ is 2. But this means $\S(AC) + \S(BC) \ngeq \S(C) + \S(ABC)$. Hence, SSA fails for this unbalanced network.

\begin{figure}[H]
    \centering
    \begin{tikzpicture}
        \node[style={draw, circle, fill=black, inner sep=2pt}] (B) at (0, 0) {};
        \node[style={draw, circle, fill=black, inner sep=2pt}] (i) at (2, 0) {};
        \node[style={draw, circle, fill=black, inner sep=2pt}] (C) at (2, 2) {};
        \node[style={draw, circle, fill=black, inner sep=2pt}] (A) at (0, 2) {};
        \node[style={draw, circle, fill=black, inner sep=2pt}] (O) at (1, -1) {};

        \node[below=2mm] at (O) {$B$};
        \node[above=2mm] at (C) {$C$};
        \node[left=2mm] at (B) {$O$};
        \node[above=2mm] at (A) {$A$};
        
        \draw[-latex, line width=1pt] (A) -- node[left] {} (B);
        \draw[-latex, line width=1pt] (B) -- node[above] {} (C);
        \draw[-latex, line width=1pt] (A) to[bend left=30] node[] {} (C);
        \draw[-latex, line width=1pt] (C) to[bend left=30] node[above] {} (A);
        \draw[-latex, line width=1pt] (C) -- node[] {} (i);
        \draw[-latex, line width=1pt] (i) -- node[above] {} (B);
        \draw[-latex, line width=1pt] (B) -- node[left] {} (O);
        \draw[-latex, line width=1pt] (O) -- node[right] {} (i);
    \end{tikzpicture}
    \caption{An unbalanced network, where all the edges have capacity 1, that does not satisfy SSA. It is straightforward to verify $\S(C) = 2, \S(AC) = 1, \S(BC) = 2,$ and $\S(ABC)=2$. Then $\S(AC) + \S(BC) 
 = 1 + 2 \ngeq \S(C) + \S(ABC) = 2 + 2$.}\label{fig:unbalanced}
\end{figure}
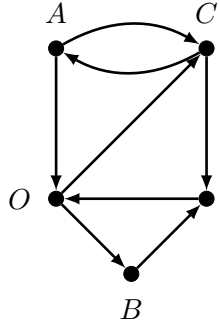

We now introduce the standard notion of a residual network given a network $\CN = (V,E,c)$, as applied to a multi-cycle flow $\mfC = \{(E_i;\l_i)\}_{i=1,\ldots,m}$. Utilizing the definition of $u_i$ given in \Cref{def:multi-flow}, the residual network is given by
\begin{align}\label{eq:residual-network}
    \Res(\CN,\mfC) \coloneqq \bigg(V,E, c - \sum_{i=1}^m u_i\lambda_i \bigg) .
\end{align}
Notice that because of \eqref{eq:multi-flow-cond}, the capacity function of the residual network is always non-negative. If the capacity of an edge in $\Res(\CN,\mfC)$ is 0, then we simply remove the edge. Intuitively, the residual network is the remaining capacity of the edges of $\CN$ given there is a multi-cycle flow $\mfC$ already flowing on $\CN$, and is useful for many reasons, including the following proposition.

\begin{prop}\label{prop:adding-flows}
    Given any (not necessarily balanced) network $\CN = (V,E,c)$ and a multi-cycle flow $\mfC$, if $\CC$ is a simple cycle flow in $\Res(\CN,\mfC)$, then $\mfC \cup \CC$ is a multi-cycle flow in the original network $\CN$.    
\end{prop}
\begin{proof}
    Denoting $\mfC = \{\CC_i = (E_i;\lambda_i) \}_{i=1,\ldots,m}$, the capacity of every edge of the residual network $\Res(\CN,\mfC)$ is given by
    \begin{align}
        c(e) - \sum_{i=1}^m u_i(e)\lambda_i,
    \end{align}
    where $u_i(e) = 1$ if $e \in E_i$ and 0 otherwise, as was introduced in \Cref{def:multi-flow}. Since $\CC = (E_\CC;\lambda)$ is a cycle flow in $\Res(\CN,\mfC)$, this means by \Cref{def:cycle-flow}
    \begin{align}
        c(e) - \sum_{i=1}^m u_i(e)\lambda_i \geq \lambda \quad\text{for every $e \in E_\CC$}.
    \end{align}
    It is trivial then to see that
    \begin{align}
        c(e) \geq u(e) \lambda + \sum_{i=1}^m u_i(e)\lambda_i,
    \end{align}
    where $u(e) = 1$ if $e \in E_\CC$ and 0 otherwise. This proves that $\mfC \cup \CC$ is a cycle flow in the original network $\CN$.
\end{proof}

Furthermore, for the special case where the network $\CN$ is balanced, we have the following nice property.

\begin{prop}\label{prop:res}
    Given a balanced network $\CN=(V,E,c)$ and a multi-cycle flow $\mfC$, the residual network $\Res(\CN,\mfC)$ is a balanced network. Furthermore, there exists a multi-cycle flow $\mfC$ such that the residual network has $c(e) = 0$ for every $e \in E$.
\end{prop}
\begin{proof}
    Because the multi-cycle flow $\mfC$ is balanced, this means for every $v \in V$, the total capacity reduction of the edges going into $v$ equals the total capacity reduction of the edges coming out of $v$, which proves the claim $\Res(\CN,\mfC)$ is balanced. 

    To prove the second part of the claim, consider some multi-cycle flow $\mfC'$ in $\CN$. It is obvious that the capacity of every edge in $\CN_1 \coloneqq \Res(\CN,\mfC')$ is less than or equal to that in $\CN$, with the capacity of at least one edge in $\CN_1$ being strictly less than that in $\CN$. If $c(e) > 0$ for some $e \in E$ in $\CN_1$, we can construct a simple cycle $\CC_1$ containing $e$ by first constructing an edge cycle $E_{\CC_1}$ explicitly starting from $e$; this is always possible because $\CN_1$ is balanced. The magnitude of $\CC_1$ is then given to be $\min_{e \in E_{\CC_1}} c(e)$. Notice this means the residual network $\CN_2 \coloneqq \Res(\CN_1,\CC_1)$ has at least one fewer edge than $\CN_1$. We can then repeat this process with $\CN_1$ and $\CN$ and $\CN_2$ in place of $\CN_1$ and obtain a simple cycle flow $\CC_2$. Repeating this until the resulting residual network has no edges, i.e. $c(e) = 0$ for all edges, it then follows $\mfC \coloneqq \mfC' \cup \CC_1 \cup \cdots \cup \CC_n$ is a multi-cycle flow where the residual network $\Res(\CN,\mfC)$ has $c(e)=0$ for every edge.
\end{proof}

We now want to specialize even further and study a particular subclass of balanced networks that satisfies a crucial property known as nesting, which is the subclass of balanced networks where for any two disjoint subset of boundary vertices $\pII,\pJJ$, there exists a multi-cycle flow $\mfC(\pII,\pJJ)$ that simultaneously maximizes the flux through both $\pII$ and $\pII\pJJ$. Not every balanced network obeys the nesting condition, and we will prove a necessary and sufficient condition for nesting to hold.\footnote{This is markedly different from the universal existence of nested flows for conventional flows involving sources and sinks (e.g., see \cite{Freedman:2016zud}). It would be interesting to further explore connections between cycle flows and conventional flows, especially in the context of multi-commodity flows.} %
As we will see in the next section, the nesting property will be repeatedly used to prove that \eqref{eq:max-flow} obeys both SA and SSA, thereby motivating us to identify it with the entropy.

\begin{theorem}\label{thm:nesting}
    Given a balanced network $\CN = (V,E,c)$ with boundary vertices $\p V \subseteq V$, and given two disjoint subsets of boundary vertices $\pII,\pJJ$ where $\pII \cap \pJJ = \oslash$, let $\pKK \coloneqq \p V \setminus \pII\pJJ$ denote the rest of the boundary vertices. Then there exists a multi-cycle flow $\mfC(\pII,\pJJ)$ with respect to the boundary vertex partition $\p V = \{\pII,\pJJ,\pKK\}$, called a \emph{nested multi-cycle flow}, whose magnitude through $\pII$ and $\pII\pJJ$ is simultaneously maximized if and only if there exists a multi-cycle flow $\mfC(\pII\pJJ)$ maximizing the flux through $\pII\pJJ$, a multi-cycle flow $\mfC(\pII)$ maximizing the flux through $\pII$, and a multi-cycle flow $\mfC$ that is a multi-cycle subflow of both $\mfC(\pII\pJJ)$ and $\mfC(\pII)$ such that $\pII$ and $\pKK$ are disjoint in $\Res(\CN,\mfC)$.
\end{theorem}
\begin{proof}
    First, let us prove the forward direction, and assume there exists a nested multi-cycle flow $\mfC(\pII,\pJJ)$ that maximizes the flux through both $\pII$ and $\pII\pJJ$ simultaneously. We can then choose $\mfC(\pII) = \mfC(\pII\pJJ) = \mfC(\pII,\pJJ)$. Then $\mfC(\pII,\pJJ)$ is trivially a multi-cycle subflow of both $\mfC(\pII)$ and $\mfC(\pII\pJJ)$, so we can also let $\mfC = \mfC(\pII,\pJJ)$. We now claim that $\pII$ and $\pKK$ are disjoint in the residual network $\Res(\CN,\mfC)$. To prove this claim, suppose the contrary and assume there is still a path between $\pII$ and $\pKK$. Since $\CN_2$ is balanced by \Cref{prop:res}, there still exists a simple cycle $\CC$ between $\pII$ and $\pKK$ (possibly including $\pJJ$) in $\Res(\CN,\mfC)$. This means we can augment $\mfC(\pII,\pJJ)$ by $\CC$ to increase the flux through $\pII$, contradicting the fact $\mfC(\pII,\pJJ)$ maximizes the flux through $\pII$. This proves $\pII$ and $\pKK$ must be disjoint, thereby proving the forward direction.

    Next, we prove the reverse direction, and suppose there exists such a multi-cycle subflow $\mfC$. In particular, this means in the residual network $\Res(\CN,\mfC)$, any vertex in $\pJJ$ that is part of a cycle with a vertex in $\pKK$ cannot be connected to any vertices in $\pII$, and similarly any vertex in $\pJJ$ that is part of a cycle involving vertices in $\pII$ cannot be connected to any vertex in $\pKK$. It follows that we can simultaneously construct multi-cycle flows $\mfC_{\pII\pJJ}$ and $\mfC_{\pJJ\pKK}$ on $\Res(\CN,\mfC)$ such that $\mfC_{\pII\pJJ}$ only involves cycles that maximize the flux between $\pII$ and $\pJJ$, and  $\mfC_{\pJJ\pKK}$ only involves cycles that maximize the flux between $\pJJ$ and $\pKK$. Thus, $\mfC \cup \mfC_{\pII\pJJ} \cup \mfC_{\pJJ\pKK}$ is a valid multi-cycle flow in the original network $\CN$, and it utilizes all the capacities of edges between any two subsets $\pII,\pJJ$, and $\pKK$.\footnote{There might be additional cycle flows within a single subset, say $\pII$, but such cycle flows do not contribute to the flux through $\pII,\pJJ$, or $\pKK$ and hence can be ignored.} It should then be clear that the flux through $\mfC(\pII\pJJ)$ equals that through $\mfC \cup \mfC_{\pJJ\pKK}$, as $\mfC_{\pII\pJJ}$ does not contribute to the flux through $\pII\pJJ$. Similarly, the flux through $\mfC(\pII)$ must equal that through $\mfC \cup \mfC_{\pII\pJJ}$, as $\mfC_{\pJJ\pKK}$ does not contribute to the flux through $\pII$. It follows that $\mfC \cup \mfC_{\pJJ\pKK} \cup \mfC_{\pII\pJJ}$ is precisely a multi-cycle flow that simultaneously maximizes the flux through both $\pII$ and $\pII\pJJ$, completing the proof in the reverse direction as well.
\end{proof}

We conclude this section by restricting ourselves to studying the balanced networks that satisfy \Cref{thm:nesting}, which we define below.
\begin{defn}\label{def:nesting}
    A balanced network $\CN = (V,E,c)$ obeys nesting if for any pair of boundary vertex subsets $\pII,\pJJ$, there exists a multi-cycle flow $\mfC(\pII,\pJJ)$ that simultaneously maximizes the flux through both $\pII$ and $\pII\pJJ$. We will call such a network a nesting balanced network.
\end{defn}
Note that not every balanced network is a nesting balanced network. An example of this is given in \Cref{fig:not-nesting}. The maximal flux through $BC$ is 3, and the maximal flux through $ABCD$ is also 3. However, the multi-cycle flow depicted in \Cref{fig:not-nesting} is the only one that maximizes the flux through $ABCD$. However, it does not maximize the flux through $BC$, and so nesting fails for this balanced network. In fact, notice that this network violates SSA, as it is straightforward to check
\begin{align}
    \S(ABC) + \S(BCD) = 3+2 \ngeq \S(BC) + \S(ABCD) = 3+3.
\end{align}
In the next section, we will prove specifically for nesting balanced networks, SA and SSA are always satisfied.

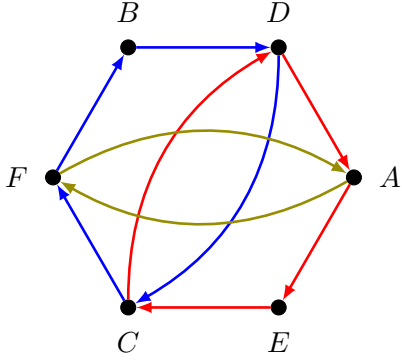
\begin{figure}[H]
    \centering
    \begin{tikzpicture}
        \node[style={draw, circle, fill=black, inner sep=2pt}] (A) at (2, 0) {};
        \node[style={draw, circle, fill=black, inner sep=2pt}] (F) at (-2, 0) {};
        \node[style={draw, circle, fill=black, inner sep=2pt}] (D) at (1, 1.73) {};
        \node[style={draw, circle, fill=black, inner sep=2pt}] (B) at (-1, 1.73) {};
        \node[style={draw, circle, fill=black, inner sep=2pt}] (E) at (1, -1.73) {};
        \node[style={draw, circle, fill=black, inner sep=2pt}] (C) at (-1, -1.73) {};

        \node[right=2mm] at (A) {$A$};
        \node[above=2mm] at (B) {$B$};
        \node[below=2mm] at (C) {$C$};
        \node[above=2mm] at (D) {$D$};
        \node[below=2mm] at (E) {$E$};
        \node[left=2mm] at (F) {$F$};
        
        \draw[-latex, line width=1pt, draw=red] (A) -- node[left] {} (E);
        \draw[-latex, line width=1pt, draw=red] (E) -- node[above] {} (C);
        \draw[-latex, line width=1pt, draw=red] (C) to[bend left=30] node[] {} (D);
        \draw[-latex, line width=1pt, draw=red] (D) -- node[above] {} (A);
        \draw[-latex, line width=1pt, draw=blue] (D) to[bend left=30] node[above] {} (C);
        \draw[-latex, line width=1pt, draw=blue] (C) -- node[] {} (F);
        \draw[-latex, line width=1pt, draw=blue] (F) -- node[] {} (B);
        \draw[-latex, line width=1pt, draw=blue] (B) -- node[] {} (D);
        \draw[-latex, line width=1pt, draw=olive] (A) to[bend left=30] node[above] {} (F);
        \draw[-latex, line width=1pt, draw=olive] (F) to[bend left=30] node[above] {} (A);
    \end{tikzpicture}
    \caption{A balanced network, where all edges have capacity 1, that does not satisfy the nesting property. We have $\S(BC) = \S(ABCD) = 3$. Each colored cycle denotes a simple cycle in the only multi-cycle flow that maximizes the flux through $ABCD$, and it does not simultaneously maximize the flux through $BC$. Note that this directed network is planar (although we did not draw it in a manifestly planar manner for aesthetic reasons), so planarity of the network does not imply nesting.}\label{fig:not-nesting}
\end{figure}

\section{Cycle Flow Entropy Cone}\label{sec:entropy-cone}

In this section, we give a prescription for using the fluxes of the cycle flows of networks to represent entropies. This is introduced in \Cref{ssec:entropy}. We will then prove that the entropies defined via this prescription satisfy SA in \Cref{ssec:sa} and SSA in \Cref{ssec:ssa} by using various properties of minimal cycle flows.

\subsection{Entropy as Cycle Flows}\label{ssec:entropy}

Given a nesting balanced network $\CN=(V,E,c)$, we designate a subset of its vertices $\p V = [\N+1] \subseteq V$, which we call boundary vertices, such that they each correspond to the party $A_i$ for $i \in [\N+1]$ in our system, with $A_{\N+1}$ being the purifier. We will then identify the entropy of the parties $\bigcup_{i\in\pII} A_i$ to be the maximum magnitude of a multi-cycle flow through $\pII \subseteq \p V$, as defined in \eqref{eq:max-flow}. In other words, recalling \eqref{eq:max-flow}, we define
\begin{align}\label{eq:entropy-def}
    \text{Entropy of $\bigg( \bigcup_{i\in\pII} A_i$} \bigg) := \S(\pII) = \max_{\mfC} f_{\mfC}(\pII,\pII^c).
\end{align}

Before we check various properties of this definition, it is useful to give some examples of how to associate an entropy vector to a nesting balanced network. First, we claim that the prescription above describes the entropy vectors of all holographic states. To see why, notice that any undirected graph $\CG = (V,E)$ with weighted edges can be represented as a balanced network such that the maximum flow out of any subset of vertices remain unchanged. For every edge $e \in E$ with capacity $c(e)$ connecting vertices $v_1$ and $v_2$, we can replace it with two directed edges such that one directed edge goes from $v_1$ to $v_2$ with capacity $c(e)$ and the other goes from $v_2$ to $v_1$ with the same capacity $c(e)$ (see \Cref{fig:GtoN}). The nesting property is then inherited from the nesting property of conventional flows. This gives rise to a nesting balanced network $\CN$, and the flow out of any vertex stays the same as that of the original graph $\CG$. 

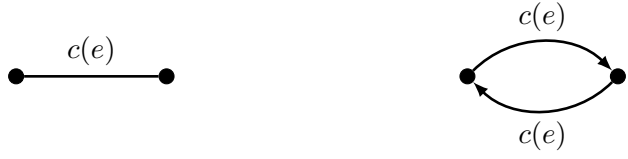
\begin{figure}[H]
    \centering
    \begin{tikzpicture}[scale=1, node distance=2cm]
        \node[style={draw, circle, fill=black, inner sep=2pt}] (A) at (0, 0) {};
        \node[style={draw, circle, fill=black, inner sep=2pt}] (B) at (2, 0) {};
        \draw[line width=1pt] (A) -- node[above] {$c(e)$} (B);
        
        \begin{scope}[xshift=6cm]
            \node[style={draw, circle, fill=black, inner sep=2pt}] (C) at (0, 0) {};
            \node[style={draw, circle, fill=black, inner sep=2pt}] (D) at (2, 0) {};
            \draw[-latex, line width=1pt] (C) to[bend left=45] node[above] {$c(e)$} (D);
            \draw[-latex, line width=1pt] (D) to[bend left=45] node[below] {$c(e)$} (C);
        \end{scope}
    \end{tikzpicture}
    \caption{The left graph represents the entropy of a Bell pair in the original graph model. We can equally represent the entropy vector using cycle flows via the network on the right.} \label{fig:GtoN}
\end{figure}

We now prove that the cycle flow model we are proposing in \eqref{eq:entropy-def} is sufficient to describe all original holographic entropy vectors describable by graph models.

\begin{prop}\label{prop:holographic-generalize}
    All holographic entropy vectors can be represented using the max cycle flux of nesting balanced networks given in \eqref{eq:entropy-def}.
\end{prop}
\begin{proof}
    Given a holographic entropy vector, it can be represented by some undirected graph $\CG=(V,E)$, with vertices $V$ and weighted edges $E$ \cite{Bao:2015bfa}. In particular, there exists a subset of vertices $\p V = [\N+1] \subseteq V$, which we call boundary vertices, and these vertices represent the $\N$ parties and the purifier. The entropy associated to any subset $\pII \subseteq \p V$ is then maximum flow out of $\pII$ to its complement $\pII^c$, and as we mentioned above such a flow still exists in the directed network constructed from $\CG$ by replacing each undirected edge with a pair of oppositely directed edges shown in \Cref{fig:GtoN}. We can easily augment the flow into a cycle flow by simply using the edges going in the opposite direction to go from $\pII^c$ back to the original vertices $\pII$. The cycle flux in this case just equals the flux out of $\pII$, so the max flux out of $\pII$ in $\CG$ equals the max cycle flux out of $\pII$ in the corresponding network, completing the proof.
\end{proof}

However, our goal is to move beyond holographic entropy vectors, and it is instructive to see an example. Consider the balanced network in \Cref{fig:GHZ3}, where the vertices are labeled by $A,B,C$, and purifier $O$, and all edges have capacity $1$.

\begin{figure}[H]
    \centering
    \begin{tikzpicture}
        \node[style={draw, circle, fill=black, inner sep=2pt}] (O) at (0, 0) {};
        \node[style={draw, circle, fill=black, inner sep=2pt}] (C) at (2, 0) {};
        \node[style={draw, circle, fill=black, inner sep=2pt}] (B) at (2, 2) {};
       \node[style={draw, circle, fill=black, inner sep=2pt}] (A) at (0, 2) {};

        \node[below=2mm] at (O) {$O$};
        \node[below=2mm] at (C) {$C$};
        \node[above=2mm] at (B) {$B$};
        \node[above=2mm] at (A) {$A$};
        
        \draw[-latex, line width=1pt] (A) -- node[above] {$e_1$} (B);
        \draw[-latex, line width=1pt] (B) -- node[right] {$e_2$} (C);
        \draw[-latex, line width=1pt] (C) -- node[below] {$e_3$} (O);
        \draw[-latex, line width=1pt] (O) -- node[left] {$e_4$} (A);
    \end{tikzpicture}
    \caption{A balanced network, where every edge implicitly has capacity 1, representing a 4-party GHZ state.}\label{fig:GHZ3}
\end{figure}
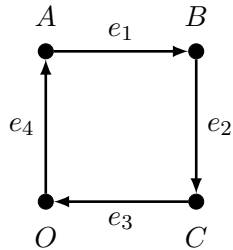

For this network, every vertex is a boundary vertex. Furthermore, given any subset of boundary vertices, the multi-cycle flow with maximum magnitude through any subset of boundary vertices is the minimal cycle flow $\CC = (\{e_1,e_2,e_3,e_4\};1)$, which is the clockwise flow with magnitude $1$. Thus, this network trivially satisfies nesting. By our definition, the magnitude of this cycle flow through all subsets of boundary vertices is
\begin{align}\label{eq:ghz3}
    \S(A) = \S(B) = \S(C) = \S(AB) = \S(AC) = \S(BC) = \S(ABC) = 1.
\end{align}
This set of entropies precisely matches those of the 4-party GHZ state. The GHZ state is not holographic, since it clearly violates MMI \eqref{eq:mmi}. Thus, we see that by identifying maximum flux associated to cycle flows as entropy, we are able to obtain entropy vectors of non-holographic states.\footnote{It was proven in Theorem 5 of \cite{Cui:2018dyq} that for directed inner-superbalanced networks (which includes balanced networks), if we identify the max flow out of a subset of boundary vertices $\pII$ as the entropy $S_\pII$, then MMI is obeyed. However, since we identify entropy $S_\pII$ not as the max flow out of $\pII$, but as the max flux of a multi-cycle flow between $\pII$ and $\pII^c$, the results of \cite{Cui:2018dyq} do not apply here. For instance, the standard max flow out of vertices $\{A,C\}$ in \Cref{fig:GHZ3} is $2$, whereas the max flux of a cycle flow through $\{A,C\}$ is $1$.}

\begin{figure}[H]
    \centering
    \begin{tikzpicture}
        \node[style={draw, circle, fill=black, inner sep=2pt}] (A) at (0, 2) {};
        \node[style={draw, circle, fill=black, inner sep=2pt}] (B) at (1.88, 0.37) {};
        \node[style={draw, circle, fill=black, inner sep=2pt}] (C) at (-1.88,0.37) {};
        \node[style={draw, circle, fill=black, inner sep=2pt}] (D) at (1.18,-1.62) {};
        \node[style={draw, circle, fill=black, inner sep=2pt}] (E) at (-1.18,-1.62) {};

        \node[above=2mm] at (A) {$A_1$};
        \node[right=2mm] at (B) {$A_2$};
        \node[left=2mm] at (C) {$A_{n}$};
        \node[below=2mm] at (D) {$A_3$};
        \node[below=2mm] at (E) {$A_{n-1}$};

        \draw[-latex, line width=1pt] (A) -- (B);
        \draw[-latex, line width=1pt] (B) -- (D);
        \draw[loosely dashed, line width=1pt] (D) -- (E);
        \draw[-latex, line width=1pt] (E) -- (C);
        \draw[-latex, line width=1pt] (C) -- (A);
    \end{tikzpicture}
    \caption{An $n$-party GHZ state, which is also representable as a hypergraph.}
    \label{fig:GHZN}
\end{figure}
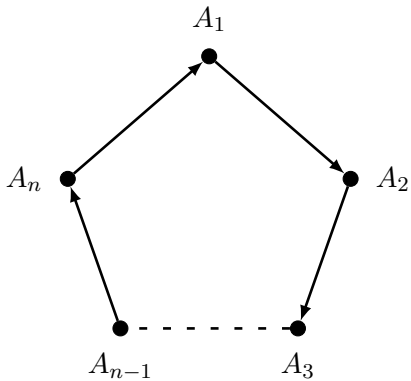

From the above example, we can easily construct a nesting balanced network corresponding to a $n$-party GHZ state. Labeling the regions $A_1,\ldots,A_{n}$,  we depict this state in \Cref{fig:GHZN}. As in the 4-party GHZ network, every vertex is a boundary vertex, and the multi-cycle flow maximizing the magnitude through any subset of boundary vertices is again the clockwise cycle flow with magnitude $1$. Thus we find
\begin{align}
    \S(\pII) = 1 \quad\text{for any $\pII \subset [n]$}.
\end{align}
Notice that this also captures the notion of a general hypergraph edge explored in \cite{Bao:2020zgx}. For instance, we can represent hypergraph 7 given in Figure 4 of \cite{Bao:2020zgx} as a cycle flow in \Cref{fig:hypergraph7} below. We simply replace each 4-vertex hyperedge with the 4-vertex cyclic network shown in \Cref{fig:GHZ3} and every undirected edge with the two directed edges shown in \Cref{fig:GtoN}. We leave it as an exercise for the reader to check that given our prescription \eqref{eq:entropy-def}, we indeed reproduce the entropy vector of the hypergraph. Furthermore, we conjecture below that in general, entropy vectors representable by hypergraphs can be captured by cycle flows.

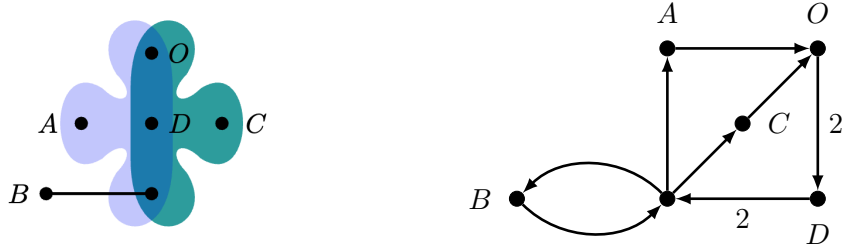
\begin{figure}[H]
    \centering
    \begin{tikzpicture}
        \node[style={draw, circle, fill=black, inner sep=2pt}] (C) at (0, 0) {};
        \node[style={draw, circle, fill=black, inner sep=2pt}] (A) at (-1, 1) {};
        \node[style={draw, circle, fill=black, inner sep=2pt}] (O) at (1,1) {};
        \node[style={draw, circle, fill=black, inner sep=2pt}] (D) at (1,-1) {};
        \node[style={draw, circle, fill=black, inner sep=2pt}] (i) at (-1,-1) {};
        \node[style={draw, circle, fill=black, inner sep=2pt}] (B) at (-3,-1) {};

        \node[above=2mm] at (A) {$A$};
        \node[left=2mm] at (B) {$B$};
        \node[right=2mm] at (C) {$C$};
        \node[below=2mm] at (D) {$D$};
        \node[above=2mm] at (O) {$O$};

        \draw[-latex, line width=1pt] (A) -- (O);
        \draw[-latex, line width=1pt] (O) -- (D) node[midway, right]{2};
        \draw[-latex, line width=1pt] (D) -- (i) node[midway, below]{2};
        \draw[-latex, line width=1pt] (i) -- (A);
        \draw[-latex, line width=1pt] (C) -- (O);
        \draw[-latex, line width=1pt] (i) -- (C);
        \draw[-latex, line width=1pt] (i) to[bend right=45] (B);
        \draw[-latex, line width=1pt] (B) to[bend right=45] (i);
    \begin{scope}[shift={(-8,0)}]    \node[] at (0,0) {\includegraphics[scale=0.22]{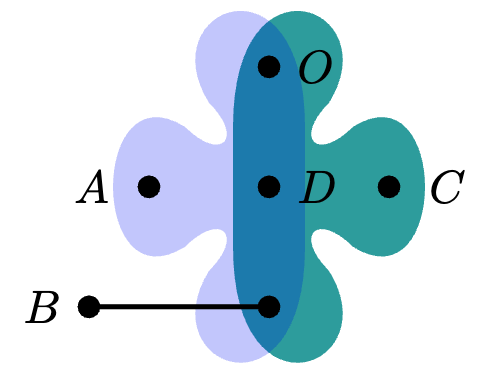}};

    \end{scope}
    \end{tikzpicture}
    
    \caption{The network on the right reproduces the entropy vector represented by hypergraph 7 in Figure 4 of \cite{Bao:2020zgx}, which we reproduced on the left for convenience. All edges have capacity $1$ unless otherwise labeled.}
    \label{fig:hypergraph7}
\end{figure}

\begin{conj}\label{conj:hypergraph-generalize}
    All entropy vectors representable by hypergraphs can be represented using the prescription given in \eqref{eq:entropy-def} on nesting balanced networks.
\end{conj}

To understand the rationale behind this conjecture, note that given any hypergraph model, we can replace every undirected edge by two oppositely directed edges shown in \Cref{fig:GtoN}, and every hyperedge involving $n$ vertices by a $n$-cyclic network like the one shown in \Cref{fig:GHZN} with the same capacity for each edge. There is still an ambiguity in the ordering of the vertices and the orientation of each such directed edge cycles replacing the hyperedges, but in many examples considered, we were able to find an orientation where the entropy vector represented by the hypergraph is reproduced by our directed network. Nevertheless, such a naive argument cannot be true in general, as the balanced network considered in \Cref{fig:not-nesting} does not have a hypergraph realization (since hypergraphs obey SSA), even though our above prescription for replacing directed edges with edges and hyperedges would suggest otherwise.\footnote{We thank Massimiliano Rota for pointing this out.} We therefore leave proving the conjecture an open question.

Although we have not proved the above conjecture, the ability for directed networks to represent entropy vectors of quantum states not representable by hypergraphs is illustrated by the balanced network given in \Cref{fig:vector15}. This network is of particular interest as the entropic pattern it realizes cannot be captured by hypergraphs. In other words, it demonstrates that the proposed cycle flow prescription generalizes beyond the hypergraph models of \cite{Bao:2020zgx} (assuming \Cref{conj:hypergraph-generalize}). The entropy vector that this network yields corresponds precisely to that of the stabilizer ``graph state 15'' of \cite{Bao:2020mqq}, which proved that the entanglement structure of such a stabilizer state is unattainable by hypergraphs. It was later shown in \cite{Bao:2021gzu} that the entropy vector of graph state 15 can be obtained by providing a min cut model on topological links while still realizing all hypergraph entropies, albeit obscuring any potentially dual flow perspective. The topological character of the models of \cite{Bao:2021gzu}, and their realization of graph state $15$, served as motivation for our construction of the state using directed networks in \Cref{fig:vector15}, as well for our cycle flow prescription for computing entropy in general. Although proving an equivalence between our prescription and that of \cite{Bao:2021gzu} at the level of entropy cones is beyond the scope of this paper, the fact that we are able to reproduce the subsystem entropies of graph state $15$ does demonstrate that cycle flows share key topological features with link models.

\begin{figure}[H] 
    \centering
    \begin{tikzpicture}
        \node[style={draw, circle, fill=black, inner sep=2pt}] (A) at (2, 2) {};
        \node[style={draw, circle, fill=black, inner sep=2pt}] (B) at (4.23, 2) {};
        \node[style={draw, circle, fill=black, inner sep=2pt}] (C) at (5.23, 0) {};
        \node[style={draw, circle, fill=black, inner sep=2pt}] (D) at (4.23, -2) {};
        \node[style={draw, circle, fill=black, inner sep=2pt}] (E) at (2.23, -2) {};
        \node[style={draw, circle, fill=black, inner sep=2pt}] (O) at (1, 0) {};
        
        \node[style={draw, circle, fill=black, inner sep=2pt}] (V1) at (2, 4) {};
        \node[style={draw, circle, fill=black, inner sep=2pt}] (V2) at (4.23, 4) {};
        \node[style={draw, circle, fill=black, inner sep=2pt}] (V3) at (7.23, -1) {};
        \node[style={draw, circle, fill=black, inner sep=2pt}] (V4) at (6.23, -3) {};
        \node[style={draw, circle, fill=black, inner sep=2pt}] (V5) at (0,-3) {};
        \node[style={draw, circle, fill=black, inner sep=2pt}] (V6) at (-1,-1) {};
        
        \node[above=2mm] at (V1) {$A$};
        \node[above=2mm] at (V2) {$B$};
        \node[right=2mm] at (V3) {$D$};
        \node[right=2mm] at (V4) {$E$};
        \node[left=2mm] at (V5) {$C$};
        \node[left=2mm] at (V6) {$O$};

        \draw[-latex, line width=1pt] (A) -- (B);
        \draw[-latex, line width=1pt] (B) -- node[right=1mm] {$2$} (C);
        \draw[-latex, line width=1pt] (C) -- (D);
        \draw[-latex, line width=1pt] (D) -- node[below=1mm] {$2$} (E);
        \draw[-latex, line width=1pt] (E) -- (O);
        \draw[-latex, line width=1pt] (O) -- node[left=1mm] {$2$} (A);  

        \draw[-latex, line width=1pt] (A) -- (V1);
        \draw[-latex, line width=1pt] (V1) -- (V2);
        \draw[-latex, line width=1pt] (V2) -- (B);
        \draw[-latex, line width=1pt] (C) -- (V3);
        \draw[-latex, line width=1pt] (V3) -- (V4);
        \draw[-latex, line width=1pt] (V4) -- (D);
        \draw[-latex, line width=1pt] (E) -- (V5);
        \draw[-latex, line width=1pt] (V5) -- (V6);
        \draw[-latex, line width=1pt] (V6) -- (O);
    \end{tikzpicture}
    \caption{This network reproduces the entropy vector represented by ``graph state 15'' in Figure 1 of \cite{Bao:2020mqq}. Its associated entropy vector is $(1 1 1 1 1 ; 1 2 2 2 2 2 2 2 2 1 ; 2 2 2 2 2 2 2 2 2 2 ; 2 2 1 2 2 ; 1)$, where we have ordered the components in lexicographic order. }
\label{fig:vector15}
\end{figure}
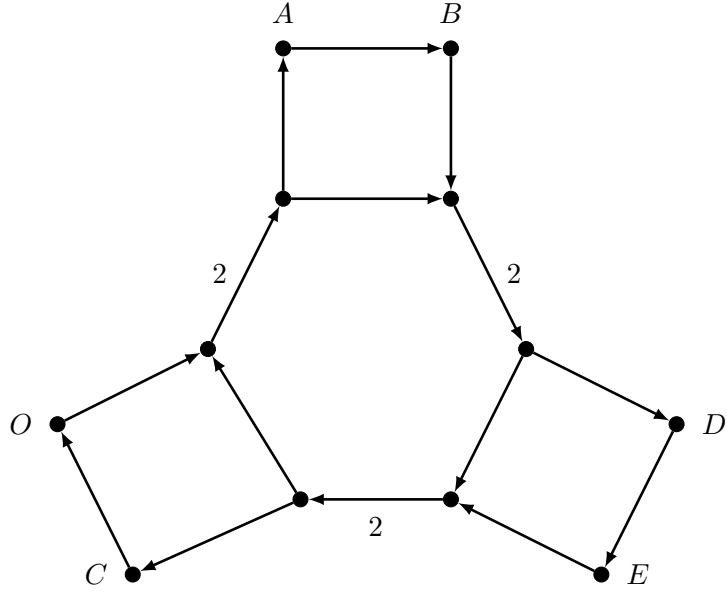

Having illustrated how to compute the maximum flux out of a subset of vertices given a cycle flow, we now turn to the properties of such max cycle flows. It is easy to see the prescription \eqref{eq:entropy-def} obeys purification symmetry.

\begin{prop}\label{lem:purification}
    Given the identification \eqref{eq:entropy-def} between entanglement entropy for a subset of vertices and the maximum flux out of the vertices in a cycle flow, purification symmetry is satisfied.
\end{prop}
\begin{proof}
    Using \eqref{eq:entropy-def}, we have
    \begin{align}
        \S(\pII) = \max_{\mfC} f_{\mfC}(\pII,\pII^c) = \max_{\mfC} f_{\mfC}(\pII^c,\pII) = \S(\pII^c),
    \end{align}
    where we simply used the fact the maximization of flux between $\pII$ and $\pII^c$ is identical to that between $\pII^c$ and $\pII$.
\end{proof}

In the next two subsections, we will show that the prescription \eqref{eq:entropy-def} also satisfies SA and SSA. Although the proofs are more involved, the spirit of the proofs is almost identical to those using bit threads found in \cite{Freedman:2016zud}.

\subsection{Subadditivity}\label{ssec:sa}

We will now prove that our prescription for entropy \eqref{eq:entropy-def} satisfies subadditivity (SA). Our approach follows closely that used in \cite{Freedman:2016zud}. 

\begin{theorem}
    The entropy defined via cycle flows on a nesting balanced network given in \eqref{eq:entropy-def} satisfies subadditivity.
\end{theorem}
\begin{proof}
    Consider a network with boundary vertices $\p V = \pII \cup \pJJ \cup \pKK$, where $\pII,\pJJ$, and $\pKK$ are mutually disjoint. We will prove
    \begin{align}
        \I(\pII:\pJJ) = \S(\pII) + \S(\pJJ) - \S(\pII\pJJ) \geq 0.
    \end{align}
    To this end, we begin by computing the conditional entropy
    \begin{align}
        \H(\pII|\pJJ) \equiv \S(\pII\pJJ) - \S(\pJJ).
    \end{align}
    To compute $\S(\pII\pJJ)$, let $\mfC(\pJJ,\pII)$ be a minimal multi-cycle flow that maximizes the magnitude through both $\pJJ$ and $\pII\pJJ$ simultaneously; such a multi-cycle flow exists by the definition of a nesting balanced network (see \Cref{def:nesting}). It then follows
    \begin{align}\label{eq:Sij}
    \begin{split}
        \S(\pII\pJJ) &= \sum_{\substack{\CC \in \mfC(\pJJ,\pII) \\ V(\CC) \cap \pJJ \neq \oslash}} f_\CC(\pII\pJJ,\pKK) + \sum_{\substack{ \CC \in \mfC(\pJJ,\pII) \\ V(\CC) \cap \pJJ = \oslash } } f_\CC(\pII\pJJ,\pKK) \\
        &= \sum_{\CC \in \mfC(\pJJ,\pII)} f_\CC(\pJJ,\pKK) + \sum_{\substack{ \CC \in \mfC(\pJJ,\pII) \\ V(\CC) \cap \pJJ = \oslash } } f_\CC(\pII,\pKK) .
    \end{split}
    \end{align}
    To obtain the second equality, we first noted that the first sum can be rewritten as follows:
    \begin{align}\label{eq:reduce}
    \begin{split}
        \sum_{\substack{\CC \in \mfC(\pJJ,\pII) \\ V(\CC) \cap \pJJ \neq \oslash}} f_\CC(\pII\pJJ,\pKK) = \sum_{\substack{\CC \in \mfC(\pJJ,\pII) \\ V(\CC) \cap \pJJ \neq \oslash}} f_\CC(\pJJ,\pKK) = \sum_{\CC \in \mfC(\pJJ,\pII)} f_\CC(\pJJ,\pKK),
    \end{split}
    \end{align}
    where the first equality follows from the fact the minimal simple cycles we are summing over all involve boundary vertices $\pJJ$ already, and the second equality follows from the fact additional minimal simple cycles included in the last sum do not involve $\pJJ$, and hence do not contribute to the sum. As for the second sum in \eqref{eq:Sij}, we used the fact it obeys the equality
    \begin{align}\label{eq:reduce2}
    \begin{split}
        \sum_{\substack{ \CC \in \mfC(\pJJ,\pII) \\ V(\CC) \cap \pJJ = \oslash }} f_\CC(\pII\pJJ,\pKK) = \sum_{\substack{ \CC \in \mfC(\pJJ,\pII) \\ V(\CC) \cap \pJJ = \oslash }} f_\CC(\pII,\pKK),
    \end{split}
    \end{align}
    since the minimal simple cycles we are summing over do not involve boundary vertices $\pJJ$.
    
    Likewise, we have
    \begin{align}
        \S(\pJJ) = \sum_{\CC \in \mfC(\pJJ)} f_\CC(\pJJ,\pII\pKK) = \sum_{\CC \in \mfC(\pJJ,\pII)} f_{\CC}(\pJJ,\pII\pKK) ,
    \end{align}
    where the second equality simply follows from the fact the minimal multi-flow $\mfC(\pJJ,\pII)$ also maximizes the flux between $\pJJ$ and $\pJJ^c = \pII\pKK$. Putting everything together, we obtain 
    \begin{align}\label{eq:cond-ent}
    \begin{split}
        \H(\pII|\pJJ) &= \sum_{\CC \in \mfC(\pJJ,\pII)} f_\CC(\pJJ,\pKK) + \sum_{\substack{ \CC \in \mfC(\pJJ,\pII) \\ V(\CC) \cap \pJJ = \oslash } } f_\CC(\pII,\pKK) - \sum_{\CC \in \mfC(\pJJ,\pII)} f_\CC(\pJJ,\pII\pKK) \\
        &= \sum_{\CC \in \mfC(\pJJ,\pII)} f_\CC(\pJJ,\pKK) + \sum_{\substack{ \CC \in \mfC(\pJJ,\pII) \\ V(\CC) \cap \pJJ = \oslash } } f_\CC(\pII,\pKK) - \left[ \sum_{\substack{ \CC \in \mfC(\pJJ,\pII) \\ V(\CC) \cap \pKK = \oslash }} f_\CC(\pJJ,\pII\pKK) + \sum_{\substack{ \CC \in \mfC(\pJJ,\pII) \\ V(\CC) \cap \pKK \neq \oslash }} f_\CC(\pJJ,\pII\pKK) \right] \\
        &= \sum_{\CC \in \mfC(\pJJ,\pII)} f_\CC(\pJJ,\pKK) + \sum_{\substack{ \CC \in \mfC(\pJJ,\pII) \\ V(\CC) \cap \pJJ = \oslash } } f_\CC(\pII,\pKK) - \left[ \sum_{\substack{ \CC \in \mfC(\pJJ,\pII) \\ V(\CC) \cap \pKK = \oslash }} f_\CC(\pJJ,\pII) + \sum_{\CC \in \mfC(\pJJ,\pII)} f_\CC(\pJJ,\pKK) \right] \\
        &= \sum_{\substack{ \CC \in \mfC(\pJJ,\pII) \\ V(\CC) \cap \pJJ = \oslash } } f_\CC(\pII,\pKK) - \sum_{\substack{ \CC \in \mfC(\pJJ,\pII) \\ V(\CC) \cap \pKK = \oslash }} f_\CC(\pJJ,\pII) , 
    \end{split}
    \end{align}
    where the second equality follows from decomposing last sum in the first line into two disjoint sums, and the third equality follows from relations analogous to \eqref{eq:reduce} and \eqref{eq:reduce2}. Using the definition of mutual information, we then have
    \begin{align}
    \begin{split}
        \I(\pII:\pJJ) &= \S(\pII) - \H(\pII|\pJJ) \\
        &=  \sum_{\CC \in \mfC(\pII,\pJJ)} f_{\CC}(\pII,\pJJ\pKK) - \left[ \sum_{\substack{ \CC \in \mfC(\pJJ,\pII) \\ V(\CC) \cap \pJJ = \oslash } } f_\CC(\pII,\pKK) - \sum_{\substack{ \CC \in \mfC(\pJJ,\pII) \\ V(\CC) \cap \pKK = \oslash }} f_\CC(\pJJ,\pII) \right] \\
        &\geq \sum_{\CC \in \mfC(\pII,\pJJ)} f_{\CC}(\pII,\pJJ\pKK) - \sum_{\substack{ \CC \in \mfC(\pJJ,\pII) \\ V(\CC) \cap \pJJ = \oslash } } f_\CC(\pII,\pKK)  \\
        &\geq \sum_{\CC \in \mfC(\pII,\pJJ)} f_{\CC}(\pII,\pJJ\pKK) - \sum_{\CC \in \mfC(\pJJ,\pII) } f_\CC(\pII,\pJJ\pKK)  \\
        &\geq 0 ,
    \end{split}
    \end{align}
    where the third line follows from the fact $f_\CC(\pJJ,\pII) \geq 0$ for all minimal simple cycles; the fourth line follows from the fact $f_\CC(\pII,\pKK) \leq f_\CC(\pII,\pJJ\pKK)$ for all minimal simple cycles, and also we are no longer restricting the sum with the condition $V(\CC) \cap \pJJ = \oslash$; and the last inequality uses the fact the first sum in the penultimate line is the maximal flux between $\pII$ and $\pII^c$ for any multi-cycle flow, whereas the second sum is the flux between $\pII$ and $\pII^c$ for the particular multi-cycle flow $\mfC(\pJJ,\pII)$. This argument completes the proof of subadditivity.
\end{proof}

\subsection{Strong Subadditivity}\label{ssec:ssa}

We now turn to prove strong subadditivity (SSA). 

\begin{theorem}
    The entropy defined via cycle flows on a nesting balanced network given in \eqref{eq:entropy-def} satisfies strong subadditivity.
\end{theorem}
\begin{proof}
    Consider a graph with boundary vertices $\p V = \pII \cup \pJJ \cup \pKK \cup \pLL$, where $\pII,\pJJ,\pKK$, and $\pLL$ are mutually disjoint. We will prove
    \begin{align}
        \I(\pII:\pJJ|\pKK) = \S(\pII\pKK) + \S(\pJJ\pKK) - \S(\pKK) - \S(\pII\pJJ\pKK) \geq 0,
    \end{align}
    where $\I(\pII:\pJJ|\pKK)$ is the conditional mutual information. First, we observe that we can write $\I(\pII:\pJJ|\pKK)$ as
    \begin{align}
        \I(\pII:\pJJ|\pKK) = \H(\pII|\pKK) - \H(\pII|\pJJ\pKK).
    \end{align}
    As our network allows for nested flows, we can use \eqref{eq:cond-ent} to obtain
    \begin{align}\label{eq:cond-mi1}
    \begin{split}
        \I(\pII:\pJJ|\pKK) &= \sum_{\substack{ \CC \in \mfC(\pKK,\pII) \\ V(\CC) \cap \pKK = \oslash } } f_\CC(\pII,\pJJ\pLL) - \sum_{\substack{ \CC \in \mfC(\pKK,\pII) \\ V(\CC) \cap \pJJ\pLL = \oslash } } f_\CC(\pKK,\pII) \\
        &\qquad - \sum_{\substack{ \CC \in \mfC(\pJJ\pKK,\pII) \\ V(\CC) \cap \pJJ\pKK = \oslash } } f_\CC(\pII,\pLL) + \sum_{\substack{ \CC \in \mfC(\pJJ\pKK,\pII) \\ V(\CC) \cap \pLL = \oslash } } f_\CC(\pJJ\pKK,\pII)  \\
        &= \sum_{\substack{ \CC \in \mfC(\pKK,\pII,\pJJ) \\ V(\CC) \cap \pKK = \oslash } } f_\CC(\pII,\pJJ\pLL) - \sum_{\substack{ \CC \in \mfC(\pKK,\pII,\pJJ) \\ V(\CC) \cap \pJJ\pLL = \oslash } } f_\CC(\pKK,\pII) \\
        &\qquad - \sum_{\substack{ \CC \in \mfC(\pKK,\pJJ,\pII) \\ V(\CC) \cap \pJJ\pKK = \oslash } } f_\CC(\pII,\pLL) + \sum_{\substack{ \CC \in \mfC(\pKK,\pJJ,\pII) \\ V(\CC) \cap \pLL = \oslash } } f_\CC(\pJJ\pKK,\pII) ,
    \end{split}
    \end{align}
    where the second equality follows from us choosing a particular minimal multi-cycle flow that maximizes the flux through $\pKK,\pII\pKK$ simultaneously in the first two sums (namely one that also maximizes the flux through $\pII\pJJ\pKK$), and similarly choosing another particular minimal multi-cycle flow that maximizes the flow through $\pJJ\pKK$ and $\pII\pJJ\pKK$ in the last two sums (namely one that also maximizes the flux through $\pKK$). Now, we observe that the first sum satisfies
    \begin{align}\label{eq:ineq1}
    \begin{split}
        \sum_{\substack{ \CC \in \mfC(\pKK,\pII,\pJJ) \\ V(\CC) \cap \pKK = \oslash } } f_\CC(\pII,\pJJ\pLL) &\geq \sum_{\substack{ \CC \in \mfC(\pKK,\pII,\pJJ) \\ V(\CC) \cap \pJJ\pKK = \oslash } } f_\CC(\pII,\pJJ\pLL) = \sum_{\substack{ \CC \in \mfC(\pKK,\pII,\pJJ) \\ V(\CC) \cap \pJJ\pKK = \oslash } } f_\CC(\pII,\pLL)  ,
    \end{split}
    \end{align}
    where the inequality follows from the fact we are further restricting the sum over minimal simple cycles, and the equality follows from the same reasoning as \eqref{eq:reduce2}. Likewise, the fourth sum in \eqref{eq:cond-mi1} satisfies
    \begin{align}\label{eq:ineq2}
    \begin{split}
        \sum_{\substack{ \CC \in \mfC(\pKK,\pJJ,\pII) \\ V(\CC) \cap \pLL = \oslash } } f_\CC(\pJJ\pKK,\pII) \geq \sum_{\substack{ \CC \in \mfC(\pKK,\pJJ,\pII) \\ V(\CC) \cap \pJJ\pLL = \oslash } } f_\CC(\pJJ\pKK,\pII) = \sum_{\substack{ \CC \in \mfC(\pKK,\pJJ,\pII) \\ V(\CC) \cap \pJJ\pLL = \oslash } } f_\CC(\pKK,\pII) .
    \end{split}
    \end{align}
    Substituting \eqref{eq:ineq1} and \eqref{eq:ineq2} into \eqref{eq:cond-mi1}, we obtain the inequality
    \begin{align}
    \begin{split}
        \I(\pII:\pJJ|\pKK) &\geq \sum_{\substack{ \CC \in \mfC(\pKK,\pII,\pJJ) \\ V(\CC) \cap \pJJ\pKK = \oslash } } f_\CC(\pII,\pLL) - \sum_{\substack{ \CC \in \mfC(\pKK,\pJJ,\pII) \\ V(\CC) \cap \pJJ\pKK = \oslash } } f_\CC(\pII,\pLL)  \\
        &\qquad  + \sum_{\substack{ \CC \in \mfC(\pKK,\pJJ,\pII) \\ V(\CC) \cap \pJJ\pLL = \oslash } } f_\CC(\pKK,\pII) - \sum_{\substack{ \CC \in \mfC(\pKK,\pII,\pJJ) \\ V(\CC) \cap \pJJ\pLL = \oslash } } f_\CC(\pKK,\pII) .
    \end{split}
    \end{align}
    Now, the first sum is the maximum flux of multi-cycle flows that are only between $\pII$ and $\pLL$ while the second sum is the minimum such flux, both subject to the constraint we are maximizing the flux through $\pKK$ and $\pII\pJJ\pKK$. The reason is because minimal simple cycle flows between $\pII$ and $\pLL$ will contribute to the maximization of the flux through $\pII\pKK$ and $\pII\pJJ\pKK$ in the first sum, whereas they only contribute to the maximization of the flux through $\pII\pJJ\pKK$ in the second sum. As the maximum cannot be less than the minimum, the first two sums combine to be non-negative. Likewise, the third sum is the maximum flux of multi-cycle flows that are only between $\pII$ and $\pKK$ while the second sum is the minimum of such flux, both subject to the constraint we are maximizing the flux through $\pKK$ and $\pII\pJJ\pKK$. Similar to before, the reason is because minimal simple cycle flows between $\pII$ and $\pKK$ contribute to the maximization of the flux through $\pJJ\pKK$ and $\pII\pJJ\pKK$ in the third sum, whereas they only contribute to the maximization through $\pII\pJJ\pKK$ in the last sum. Thus the last two sums also combine to be non-negative. Thus,
    \begin{align}\label{eq:ssa-final}
        \I(\pII:\pJJ|\pKK) \geq 0,
    \end{align}
    and so SSA is satisfied.
\end{proof}

\section{Discussion}\label{sec:discussion}

In this paper, we introduced a potential way of extending the map between undirected graph models and holographic entropy vectors to more general quantum entropy vectors. In the well-established holographic case, the undirected graph model has a clear geometric interpretation on a fixed time slice of the corresponding bulk AdS geometry \cite{Bao:2015bfa}. The boundary vertices correspond to boundary subregions of the CFT, the internal vertices correspond to bulk subregions carved out by the RT surfaces, and each edge corresponds to the portion of the RT surface dividing the two bulk subregions, with the capacity of the edge given by the area of that portion. It is clear then that the entanglement entropy associated some subset of boundary regions is given by the min cut associated to the corresponding boundary vertices in the graph. Utilizing the max flow min cut theorem, \cite{Freedman:2016zud} demonstrated how the entanglement entropy can equivalently be recast as max flow out of the corresponding boundary vertices.

In the generalization that we proposed in this paper, we do not have such a clean geometric picture. Instead of generalizing the notion of graphs to undirected hypergraphs and links, as was previously explored in \cite{Bao:2020zgx,Bao:2021gzu}, we show that it suffices to consider max flows on directed graphs. However, rather than considering traditional flows with sources and sinks, we need to consider cycle flows, as this was crucial to ensure that purification symmetry is manifest. Utilizing cycle flows, we are able to show that our prescription for computing entanglement entropy satisfies both SA and SSA. Furthermore, we straightforwardly showed that the conventional graph model representing holographic entropy states is a special case of our cycle flow model, and similarly conjectured that the hypergraph model is a special case of our cycle flow model as well.\footnote{It would be interesting to also attempt to connect our cycle flow model to the hyperthreads considered in \cite{Harper:2021uuq,Harper:2022sky}.} Importantly, the entropy vector associated to the stabilizer state that could not be reproduced by the hypergraph model has a very simple realization in our cycle flow model (see \Cref{fig:vector15}).

We would still like to obtain some form of geometric interpretation for the cycle flow model. The usual intuition that undirected edges capture geometric connectivity seems subtle to extend to directed edges. However, the work of \cite{Chandra:2022fwi} did successfully provide a geometric picture for the GHZ-like entanglement that the hypergraphs of \cite{Bao:2020zgx} capture, even though hyperedges are at first glance not geometric in the sense standard edges are.\footnote{Indeed, what geometrizes hyperedges are patched Euclidean geometries which meet at a space of higher codimension. As such, these are non-locally Euclidean spaces (thus not manifolds). We thank Tom Hartman for pointing out this geometric interpretation of GHZ-like entropies.} %
Perhaps one particularly tantalizing clue is that the entropies of GHZ states involving $n$ parties are represented simply by a directed loop involving $n$ vertices in our cycle flow model. This clue, however, also presents a new mystery: the ordering of the vertices in the directed loop is arbitrary. Any ordering produces the same entropy vector. We regard this arbitrariness as a feature -- since many states can have the same GHZ-like entropy vector, perhaps our directed networks actually contain information beyond the entropy vector.

Furthermore, in the conventional ``bit thread'' model involving max flows, one only cared about the thread configuration making up the final multi-commodity flow, and not the individual flows. However, as is evidenced by \Cref{fig:example1}, for us the simple cycles making up a multi-cycle flow are actually important, and two multi-cycle flows with exactly the same flow along each edge actually correspond to different entropy vectors depending on their respective simple cycles composing them. Additionally, the individual simple cycles composing the multi-cycle flows are directed. For an undirected graph, we can transform it into a directed graph by taking every undirected edge and replacing it with two oppositely directed edges, as is illustrated in \Cref{fig:GtoN}, and the cycle flow is obtained from the max flow traversing some set of edges and coming back along the corresponding oppositely directed edges. However, for general balanced networks, the path we take ``going there'' can be different from the one we take ``coming back.'' It would be both interesting and worthwhile to understand the physical significance of these minimal cycle flows.

Indeed, one possible avenue towards better understanding the cycle flows is an improved understanding of the dual program. Because the primal variables are the flows along the cycles (and not the magnitude of the flows along edges), the dual variables, which are variables for the edges, do not result in a simple edge-cut prescription. Although we were successful in setting up the dual program in \Cref{app:dual-program} and used it to provide alternative proofs of SA, we leave using it to find an alternative proof for SSA, as well as more generally a construction of such an edge-cut protocol, for future work. Indeed, the dual program set up in \Cref{app:dual-program} does not require any restrictions coming from considering only nesting balanced networks, so it has the added benefit of proving all balanced networks satisfy SA. We hope that this will also inform us whether there is a cleaner way to categorize the types of balanced networks that satisfy nesting besides \Cref{thm:nesting}, since while \Cref{thm:nesting} gives necessary and sufficient conditions for nesting to be satisfied, it is not straightforward to check that a random balanced network satisfies nesting.

Lastly, it is important to point out that we have only proved necessary, but not sufficient, conditions for the cycle flow model to describe quantum entropy vectors. It is possible that there exist entropy vectors that obey all instances of SA and SSA, and yet are not realizable (or approximated arbitrarily well) by any quantum state. In other words, if we define an entropy cone given by the entropy vectors realizable by cycle flows, it is not clear whether or not this cone coincides with the quantum entropy cone, contains the quantum entropy cone, or neither. The relationship between the entropy cone defined by cycle flows and that defined by either hypergraphs \cite{Bao:2020zgx} or links \cite{Bao:2021gzu} is also unknown; although we conjectured in \Cref{conj:hypergraph-generalize} that the hypergraph cone is a subset of the cycle flow cones, proving this is an important next step. We leave the complete exploration of the interplay between the entropy vectors realizable using our cycle flow model and the quantum entropy cone to future work.

\section*{Acknowledgements}

We would like to thank Nina Anikeeva, Newton Cheng, Tom Hartman, Matthew Headrick, Oliver Kosut, Shinsei Ryu, Vincent Su, Brian Swingle, and Tadashi Takayanagi for helpful conversations, and especially Massimiliano Rota for useful feedback on the preprint. We would further like to thank Ning Bao, Sebastian Fischetti, and William Munizzi, who were involved in earlier iterations of this work. T.H. would like to dedicate this to the memory of his father Yuanzhan He, whose unwavering support made it possible for T.H. to be a better physicist and, more importantly, a better person. T.H. is supported by the Heising-Simons Foundation ``Observational Signatures of Quantum Gravity'' collaboration grant 2021-2817, the U.S. Department of Energy grant DE-SC0011632, and the Walter Burke Institute for Theoretical Physics. S.H.-C. is supported by the U.S. Department of Energy Award DE-SC0021886. C.K. is supported by the U.S. Department of Energy under grant number DE-SC0019470 and by the Heising-Simons Foundation ``Observational Signatures of Quantum Gravity'' collaboration grant 2021-2818.

\appendix

\section{Remarks on the Dual Program}
\label{app:dual-program}

The max cycle flow prescription we proposed in \eqref{eq:entropy-def} for computing entanglement entropy is a linear program. Consequently, there exists a dual program. As we shall see, the dual program cannot be simply described via a min edge cut prescription; nevertheless, we will find it provides a complementary approach to the max flow prescription. In \Cref{app:linear-program}, we will set up the linear program and its dual. Next, in \Cref{app:vertex-cut}, we explain the analogous notion of a vertex cut. Lastly, we demonstrate in \Cref{app:dual-sa} how to use the dual program to provide alternative proofs of SA. Because our dual program does not need the added assumption of a nesting balanced network, our proof works for balanced networks in general. We assume reasonable familiarity with convex duality, and refer the reader to \cite{boyd:2004} for an excellent exposition on this subject. 

\subsection{The Linear Program and its Dual}\label{app:linear-program}

To construct a linear program, we need to specify a set of variables, a set of linear constraints, and the extremization goal. In this subsection, we will assume our goal is to compute $\S(\pII)$, but the generalization to multipartite regions $\pII\pJJ$ for any $\pJJ \subset \pII^c$ should be clear. 

We begin by constructing the primal program. Consider a balanced directed network $\CN = (V,E,c)$ with boundary vertices $\p V \subseteq V$ partitioned as $\p V = \{\pII,\pII^c\}$. Consider all possible edge cycles $E_i$ in $\CN$,%
\footnote{We remind the reader an edge cycle $E_i$ is an ordered sequence of distinct edges $\{e_1, e_2, \ldots, e_n \}$ such that $t(e_i) = s(e_{i+1})$ for $i=1,\ldots,n$ with $e_{n+1} \coloneqq e_1$ (see \Cref{def:cycle-flow}).} %
for $i=1,\ldots,M_\pII$, such that each $E_i$ contains both a vertex in $\pII$ and a vertex in $\pII^c = \p V \setminus \pII$. In particular, we are \emph{not} restricting our list of edge cycles to those corresponding to \emph{minimal} simple cycles. Our primal program for the flow is then given by the following:
\begin{itemize}
    \item Our variables are $\lambda_i$ for $i=1,\ldots,M_\pII$, which are the fluxes associated to each edge cycle $E_i$. 
    \item Our primal objective is to maximize
    \begin{equation}\label{eq:primal-goal}
        \sum_{i=1}^{M_\pII} \lambda_i .
    \end{equation}
    \item The linear constraints in the program are as follows. First, we have
    \begin{align}\label{eq:positive-constraint}
        \lambda_i\geq 0 \qquad\text{for all $i=1,\ldots,M_\pII$,}    
    \end{align}
    since the flux through any edge cycle is non-negative. Furthermore, for every edge $e \in \CN$, there is the capacity constraint \eqref{eq:multi-flow-cond}, which states
    \begin{equation}\label{eq:primal-constraint}
        \sum_{i=1}^{M_\pII} u_i(e) \lambda_i \leq c(e) \qquad\text{for all $e \in E$,}  
\end{equation}
where $u_i(e)=1$ if $e$ is an edge contained in the edge cycle $i$ and $0$ otherwise, and $c(e)$ is the capacity of the edge.
\end{itemize}

Notice that although this program appears very similar to the program of finding the max flow of a discrete directed graph, there is a crucial difference. In the usual max flow program, the variables are the fluxes through each edge. Accordingly, given a multi-flow consisting of flows $\CF_1,\CF_2,\ldots$, it does not matter which $\CF_i$ is contributing to the flux. On the other hand, for us, knowing the flux through each edge is not sufficient, as is evidenced in \Cref{fig:example1}, where both multi-cycle flows has exactly the same flux through the edges. Thus, we must keep track of not just the fluxes through each edge, but rather the fluxes associated to the \emph{particular cycles} themselves. As we shall see below, it is precisely this complication that renders the dual program more nontrivial than the usual min cut program.\footnote{Nevertheless, we note the rewriting of our primal program in terms of all possible cycles on the network $\CN$ is similar to the rewriting of the bit threads picture in terms of all possible paths from one region of the boundary back to another, as per Section 5.1.1 of \cite{Headrick:2022nbe}, although our cycles are directed and discrete. We leave an exploration of possible connections for future work. \label{fn:bit-thread-connection}}

We now want to construct the corresponding dual program, which amounts to switching the variables and the constraints.  In our case, the number of edge cycles $M_\pII$ in the network $\CN$ corresponds to the number of variables in the primal program, and thus to the number of constraints in the dual.  Similarly, if we treat the constraints \eqref{eq:positive-constraint} as implicit constraints, the number of explicit constraints \eqref{eq:primal-constraint} is the number of edges $e$ in the network $\CN$, which corresponds to the number of variables in the dual program.  We will denote these dual variables $d(e)$, and the dual program is given as follows:
\begin{itemize}
    \item The dual variables, one per edge, are $d(e)$.  
    \item The dual objective is to \emph{minimize} 
    \begin{equation}\label{eq:MinGoal}
        \sum_{e \in E} c(e)d(e).
    \end{equation}
    \item As for the case of the primal linear constraints, there is the dual non-negativity constraint
    \begin{align}
        d(e) \geq 0 \qquad\text{for all $e \in E$.}
    \end{align}
    Furthermore, the linear constraints in the dual program are
    \begin{equation}\label{eq:dual-constraints}
        \sum_{e \in E} u_i(e) d(e) \geq 1 \quad\text{for $i=1,\ldots,M_\pII$.}
    \end{equation}
    Thus, there are $M_\pII$ constraints, as expected.
\end{itemize}

We finish by observing that the dual program is precisely equivalent to the primal program by strong duality, which follows from Slater's condition. To check that Slater's condition is satisfied, we simply need to show the existence of primal variables such that the constraints \eqref{eq:positive-constraint} and \eqref{eq:primal-constraint} are strictly satisfied. A simple choice is to set $\lambda_i = \ve$ for a sufficiently small $\ve > 0$ for all $i = 1, \ldots, M_\pII$.

\subsection{Vertex Cuts}\label{app:vertex-cut}
As we remarked earlier, the fact our primal variables are the fluxes on cycle flows, and not edges, makes it difficult to interpret the dual variables $d(e)$ as edge cuts in the usual ``min cut'' perspective.%
\footnote{However, as is the case with standard source-sink directed flow min cuts, we suspect there exists a feasible minimization of \eqref{eq:MinGoal} where $d(e)$ is either 0 or 1. Intuitively, if we view every simple cycle as a rubber band, whose magnitude corresponds to the ``thickness'' of the rubber band, setting $d(e)=1$ corresponds to cutting enough rubber bands at edge $e$ so that the thickness of the cut rubber bands equals the capacity of the edge. In this sense, the dual variables $d(e)$ are ``cycle cuts.''} 
Instead, we will introduce the notion of a vertex cut, which we define as follows. For any network $\CN = (E,V,c)$, a vertex cut $\CV$ is a subset of (not necessarily boundary) vertices, with its magnitude given by
\begin{align}\label{eq:vertex-cut-mag}
    |\CV| \coloneqq \max_{\mfC} f_{\mfC}\big(\CV,V \setminus \CV \big).
\end{align}
Note that as defined, we have $|\CV| = |V \setminus \CV|$. Now, given some subset of boundary vertices $\pII \in \p V$, we will consider the set of all vertex cuts that separate $\pII$ from the rest of the boundary vertices $\pII^c = \p V \setminus \pII$. In other words, a vertex cut associated to $\pII$ can contain any number of internal vertices in addition to $\pII$, as long as no further boundary vertices are included. We will then show that $\S(\pII)$ is given by the minimum magnitude of all such cuts. 

For every multi-cycle flow $\mfC(\pII)$ maximizing the flux through boundary vertices $\pII$, there exists a vertex cut $\CV_{\pII}$ whose magnitude is precisely the max flow $\S(\pII)$, which we shall construct as follows. From $\mfC(\pII)$, we can construct a reduced multi-cycle flow $\mfC^r(\pII)$ by removing any simple cycle flow $\CC \in \mfC(\pII)$ that does not contain a vertex in both $\pII$ and $\pII^c$. Notice that $\mfC^r(\pII)$ still maximizes the flow out of the $\pII$ vertices; it just does not contain any extraneous simple cycle flows. Using $\mfC^r(\pII)$, we can construct the residual network $\Res(\CN,\mfC^r(\pII))$, as defined in \eqref{eq:residual-network}. 

In the residual network $\Res(\CN,\mfC^r(\pII))$, the vertices in $\pII$ are disconnected from those in $\pII^c$. Otherwise, either the network $\CN$ is unbalanced, or the flux through the multi-cycle flow $\mfC^r(\pII)$ could be increased, which is not possible as it is a max cycle flow through $\pII$. Consequently, we associate to $\mfC(\pII)$ the vertex cut $\CV_{\mfC(\pII)}$ consisting of $\pII$ and any other vertices (both internal and boundary) still connected to a vertex in $\pII$ in the residual network $\Res(\CN,\mfC^r(\pII))$. As for the magnitude $|\CV_{\mfC(\pII)}|$ of the cut, note that by \eqref{eq:vertex-cut-mag}, it cannot be simply the sum of the capacities of all edges leaving the cut, as the standard ``min cut'' perspective would suggest.%
\footnote{Consider the 4-party GHZ state in \Cref{fig:GHZ3}. If we choose our vertex cut to consist of vertices $\{A,C\}$, then by \eqref{eq:vertex-cut-mag} $|\CV(\{A,C\})| = 1$. However, if we just summed the capacity of the edges which leave the vertex cut, the flow would instead be $2$.} %
Instead, by \eqref{eq:vertex-cut-mag}, we see that
\begin{equation}\label{eq:CutValue}
    |\CV_{\mfC(\pII)}| = \max_{\mfC} f_\mfC\big( \CV_{\mfC(\pII)}, V \setminus \CV_{\mfC(\pII)} \big). 
\end{equation}
Notice that this value coincides with the magnitude of the max flow  $f_{\mfC(\pII)}(\pII,\pII^c)$, since it was built from the residual network $\Res(\CN,\mfC^r(\pII))$. Indeed, if $|\CV_{\mfc(\pII)}|$ were larger than $f_{\mfC(\pII)}(\pII,\pII^c)$, then the magnitude of the original multi-cycle flow $\mfC(\pII)$ could have been increased, which is impossible since it is a max flow arrangement. This proves $\S(\pII)$ is the magnitude of one such vertex cut containing $\pII$, and we denote this vertex cut $\CV_\pII$.

To see that $\S(\pII)$ is the minimum magnitude of all cuts containing only the boundary vertices $\pII$ (but may contain additional internal vertices), suppose $\CV'_\pII$ is another such vertex cut. If the vertex set $\CV'_\pII$ is disconnected from its complement $V \setminus \CV'_\pII$ in $\Res(\CN,\mfC^r(\pII))$, then $\mfC^r(\pII)$ is one such flow that we are maximizing over in \eqref{eq:vertex-cut-mag} when determining $|\CV'_\pII|$. Otherwise, we have to append additional simple cycles to $\mfC^r(\pII)$, obtaining a resultant multi-cycle flow $\mfC'(\pII)$, to ensure $\CV'_\pII$ is disconnected from its complement in the residual graph $\Res(\CN,\mfC'(\pII))$. In this case, $\mfC'(\pII)$ is one such flow that we are maximizing over in \eqref{eq:vertex-cut-mag} when determining $|\CV'_\pII|$, and note that by construction
\begin{align}
    f_{\mfC'(\pII)} (\CV'_\pII , V \setminus \CV'_\pII) \geq f_{\mfC(\pII)} (\CV_\pII , V \setminus \CV_\pII).
\end{align}
Thus, we see that $|\CV'_\pII| \geq |\CV_\pII|$, proving the claim.

\subsection{Alternative Proofs of SA }\label{app:dual-sa}

In this subsection, we will give two related ways to prove SA. First, we will prove it using the above described vertex cuts. Then, we will show how we can also prove SA by directly using the dual variables $d(e)$. For simplicity, we will use $\CV(\pII)$ throughout to refer to the vertex cut constructed using the max flow $\mfC(\pII)$ described above, and keep the dependence on the multi-cycle flow $\mfC(\pII)$ implicit. Furthermore, in this subsection only, we will adopt the notation $\pII^c \coloneqq V \setminus \pII$, so that the superscript $c$ indicates the complement with respect to \emph{all} the vertices, and not just the boundary vertices.

Let $\pII,\pJJ$ be two disjoint subsets of the boundary vertices, and construct the corresponding  vertex cuts $\CV(\pII)$ and $\CV(\pJJ)$ using the max flow procedure. As we argued in the previous subsection, by strong duality the magnitudes of the vertex cuts are equal to $\S(\pII)$ and $\S(\pJJ)$, respectively, so that
\begin{align}
    \S(\pII) + \S(\pJJ) = |\CV(\pII)| + |\CV(\pJJ)| .
\end{align}
Now, $\CV(\pII) \cup \CV(\pJJ)$ is obviously a vertex cut for $\pII\pJJ$. Furthermore, every simple cycle between $\CV(\pII) \cup \CV(\pJJ)$ and $(\CV(\pII) \cup \CV(\pJJ))^c$ is either a simple cycle between $\CV(\pII)$ and $\CV(\pII)^c$ or one between $\CV(\pJJ)$ and $\CV(\pJJ)^c$. Therefore, given a max multi-cycle flow $\mfC_{\pII\pJJ}$ maximizing the flux through $\CV(\pII) \cup \CV(\pJJ)$, we can rewrite this as
\begin{align}\label{eq:decomp}
    \mfC_{\pII\pJJ} = \big\{\mfC_\pII, \mfC_\pJJ\big\},
\end{align}
where $\mfC_\pII$ is the set of simple cycles in $\mfC_{\pII\pJJ}$ that involves $\CV(\pII)$ and $\CV(\pII)^c$, and $\mfC_\pJJ$ are the remaining simple cycles in $\mfC_{\pII\pJJ}$ that necessarily involves $\CV(\pJJ)$ and $\CV(\pJJ)^c$. It follows
\begin{align}
\begin{split}
    |\CV(\pII) \cup \CV(\pJJ)| &= f_{\mfC_\pII} \big(\CV(\pII) \cup \CV(\pJJ),  (\CV(\pII) \cup \CV(\pJJ))^c \big) + f_{\mfC_\pJJ} \big(\CV(\pII) \cup \CV(\pJJ),  (\CV(\pII) \cup \CV(\pJJ))^c \big) \\
    &= f_{\mfC_\pII}\big(\CV(\pII), \CV(\pII)^c \big) + f_{\mfC_\pJJ}\big(\CV(\pJJ), \CV(\pJJ)^c \big) \\
    &\leq |\CV(\pII)| + |\CV(\pJJ)| ,
\end{split}
\end{align}
where the second equality followed from the way we grouped the simple cycles in \eqref{eq:decomp}, and the last inequality follows from the fact the magnitude $|\CV(\pII)|$ is the max flow out of $\CV(\pII)$, and likewise for $|\CV(\pJJ)|$. Further using the fact that the vertex cut $\CV(\pII\pJJ)$ is the one that has the minimum magnitude among all possible vertex cuts of $\pII\pJJ$, implying that $|\CV(\pII\pJJ)| \leq |\CV(\pII) \cup \CV(\pJJ)|$, we thus have
\begin{align}
    \S(\pII) + \S(\pJJ) = |\CV(\pII)| + |\CV(\pJJ)| \geq |\CV(\pII) \cup \CV(\pJJ)| \geq |\CV(\pII\pJJ) | = \S(\pII\pJJ),
\end{align}
thereby giving our first alternative proof of SA.

Similarly, we can also prove SA directly using the dual variables $d(e)$. Let $d_I(e)$ and $d_J(e)$ be the dual variables chosen to minimize the dual objectives \eqref{eq:MinGoal} corresponding to $\pII$ and $\pJJ$ while satisfying the constraints \eqref{eq:dual-constraints}, so that
\begin{align}
\begin{split}
    \sum_{e \in E} u_i^I(e) d_I(e) &\geq 1   \quad\text{for $i = 1,\ldots, M_\pII$} \\
    \sum_{e \in E} u_i^J(e) d_J(e) &\geq 1   \quad\text{for $i = 1,\ldots, M_\pJJ$},    
\end{split}  
\end{align}
where $u_i^I(e) = 1$ if $e$ is contained in the $i$th cycle between $\pII$ and $\pII^c$ and 0 otherwise, and likewise for $u_i^J$. It is then obvious that by the non-negativity of the dual variables $d(e)$ that
\begin{align}
\begin{split}
    \sum_{e \in E} u_i^{I}(e) (d_I(e) + d_J(e)) \geq 1 \quad\text{for $i = 1,\ldots, M_\pII$} \\
    \sum_{e \in E} u_i^{J}(e) (d_I(e) + d_J(e)) \geq 1 \quad\text{for $i = 1,\ldots, M_\pJJ$} .
\end{split}
\end{align}
Since every simple cycle between $\pII\pJJ$ and $(\pII\pJJ)^c$ must be a simple cycle between $\pII$ and $\pII^c$ or one between $\pJJ$ and $\pJJ^c$, this means $d_I(e) + d_J(e)$ is a dual variable associated to the simple cycle flows between $\pII\pJJ$ and $(\pII\pJJ)^c$. Thus, by the dual program, if we denote $d_{IJ}(e)$ to be the dual variable that minimizes \eqref{eq:MinGoal} given the corresponding constraint \eqref{eq:dual-constraints} involving cycle flows between $\pII\pJJ$ and $V \setminus (\pII\pJJ)$, we get
\begin{align}
\begin{split}
    \S(\pII) + \S(\pJJ) = \sum_{e \in E} c(e) (d_I(e) + d_J(e)) \geq \sum_{e \in E} c(e) d_{IJ}(e) = \S(\pII\pJJ),
\end{split}
\end{align}
thereby proving SA in yet another way.

Of course, we are also interested to prove SSA using the dual program. As we proved that SSA holds only for nesting balanced networks, this would involve understanding how to capture the notion of nesting cycle flows from the dual perspective. We leave such explorations to future work.

\bibliography{cycle-flow-bib}{}
\bibliographystyle{utphys}

\end{document}